\documentclass[11pt]{article}

\usepackage[letterpaper,margin=1.00in]{geometry}

\usepackage[colorlinks=true,citecolor=blue,urlcolor=blue,bookmarks=false]{hyperref}
\usepackage[numbers,sort]{natbib}
\usepackage{amsfonts} 
\usepackage{amsmath,amsthm,amssymb}
\usepackage{thm-restate}
\usepackage[nameinlink]{cleveref}
\usepackage{xcolor}
\usepackage{bm}
\usepackage{xspace}
\usepackage[textsize=tiny]{todonotes}
\usepackage{hyphenat} 

\usepackage{algorithm}
\usepackage{algorithmic}

\newlength\myindent
\newcommand\bindent{%
  \begingroup
  \setlength{\itemindent}{\myindent}
  \addtolength{\algorithmicindent}{\myindent}
}
\newcommand\eindent{\endgroup}

\usepackage[framemethod=TikZ]{mdframed}

\theoremstyle{plain}
\newtheorem{theorem}{Theorem}[section]
\newtheorem{corollary}[theorem]{Corollary}

\newtheorem{lemma}[theorem]{Lemma}

\newtheorem{definition}[theorem]{Definition}

\newcommand{\eps}{\varepsilon}
\newcommand{\con}{\mathsf{congestion}}
\newcommand{\dil}{\mathsf{dilation}}

\newcommand{\E}{\mathbb{E}}

\newcommand{\calT}{\mathcal{T}}
\newcommand{\calD}{\mathcal{D}}
\newcommand{\calP}{\mathcal{P}}

\newcommand{\dist}{d}

\newcommand{\at}[1]{^{(#1)}}

\newcommand{\1}[1]{\mathbb{I}[#1]}
\newcommand{\opt}{\mathrm{opt}}

\renewcommand{\cong}{\mathrm{cong}}

\newcommand{\poly}{\mathrm{poly}}
\newcommand{\hop}{\mathrm{hop}}

\newcommand{\set}[1]{{\{#1\}}}

\DeclareMathOperator{\supp}{supp}
\DeclareMathOperator{\flow}{flow}

\newcommand{\rev}[1]{#1}

\begin{document}
\title{Hop-Constrained Oblivious Routing\footnote{Supported in part by Swiss National Foundation (project grant 200021-184735), NSF grants CCF-1527110, CCF-1618280, CCF-1814603, CCF-1910588, NSF CAREER award CCF-1750808, a Sloan Research Fellowship, and funding from the European Research Council (ERC) under the European Union's Horizon 2020 research and innovation program (ERC grant agreement No. 853109 and grant agreement No. 949272).}}

\author{
  Mohsen Ghaffari\\
  \small MIT \\
  \small ghaffari@mit.edu
  \and	
  Bernhard Haeupler \\
  \small ETH Zurich \& Carnegie Mellon University\\
  \small bernhard.haeupler@inf.ethz.ch
  \and
  Goran Zuzic \\
  \small ETH Zurich \\
  \small goran.zuzic@inf.ethz.ch
}

\date{}
\maketitle

\begin{abstract} 
We prove the existence of an oblivious routing scheme that is $\mathrm{poly}(\log n)$-competitive in terms of $(congestion + dilation)$, thus resolving a well-known question in oblivious routing. 

Concretely, consider an undirected network and a set of packets each with its own source and destination. The objective is to choose a path for each packet, from its source to its destination, so as to minimize $(congestion + dilation)$, defined as follows: The dilation is the maximum path hop-length, and the congestion is the maximum number of paths that include any single edge. The routing scheme obliviously and randomly selects a path for each packet independent of (the existence of) the other packets. Despite this obliviousness, the selected paths have $(congestion + dilation)$ within a $\mathrm{poly}(\log n)$ factor of the best possible value. More precisely, for any integer hop-constraint $h$, this oblivious routing scheme selects paths of length at most $h \cdot \mathrm{poly}(\log n)$ and is $\mathrm{poly}(\log n)$-competitive in terms of $congestion$ in comparison to the best possible $congestion$ achievable via paths of length at most $h$ hops. These paths can be sampled in polynomial time.

This result can be viewed as an analogue of the celebrated oblivious routing results of R\"{a}cke [FOCS 2002, STOC 2008], which are $O(\log n)$-competitive in terms of $congestion$, but are not competitive in terms of $dilation$.
\end{abstract}

\setcounter{page}{0}
\thispagestyle{empty}
\newpage

\section{Introduction and Related Work}
Routing packets in computer networks is a fundamental task and a widely studied problem. Consider the following prototypical scenario:

\begin{center}
  \begin{minipage}{0.85\textwidth}    
    \emph{The network is abstracted as an $n$-node undirected graph. Each edge $e$ has the capacity to transfer $c_e$ packets per time unit, and each packet traversing each edge takes one time unit. The network receives a number of packet delivery requests, where the $i^{th}$ packet should be transmitted from source $s_i$ to destination $t_i$. The objective is to minimize the packet delivery completion time, i.e., to deliver all the packets to their destinations in the shortest span of time possible.}
  \end{minipage}
\end{center}
\smallskip

This paper's contribution can be informally summarized as follows: we present the first \emph{oblivious routing} scheme that is \emph{competitive} in \emph{completion time}. The scheme is oblivious in the sense that for each packet delivery request $i$, the path chosen for this packet\footnote{We note that besides this path selection obliviousness, even the timing schedule of how the packet traverses this path is essentially independent of all other packets in the following sense: we can break time into phases each involving $\Theta(\log n)$ time units, and the phase number in which the packet traverses through any edge on its chosen path is independent of all the other packets.} is decided independent of all the other packets $j\neq i$. This obliviousness property is strongly desirable in numerous networking settings, where packet delivery requests arrive at various points in the network and their routing has to be determined without any central control of the state of the network. The competitiveness guarantee is that, albeit this obliviousness restriction, all the packets are delivered to their destinations within a time that is at most a $\poly(\log n)$ factor larger than the optimal time that is needed to deliver all the packets, i.e., the completion time in the fastest possible way to deliver all the packets.

In what follows, we describe this contribution in a more formal manner, putting it in the context of what has been known about packet routing algorithms in networks and especially the prior results on oblivious routing.

\subsection{Background on routing}
\paragraph{Route selection and scheduling:} The problem of routing a set of packets from their sources $s$ to destinations $t$ while minimizing the completion time involves two components: (I) \emph{route selection}, i.e., choosing the path $p_{s,t}$ along which the packet is transferred from $s$ to $t$, (II) \emph{scheduling} the timing of the packet traversing this path $p_{s,t}$, i.e., at which time unit the packet goes through each edge $e\in p_{s,t}$. A celebrated result of Leighton, Maggs, and Rao~\cite{leighton1994packet} shows that one can decouple these two issues with only a moderate loss---once the routes are selected, we can solve scheduling nearly optimally. Let us make this more precise. We focus on the setting where all edge capacities are uniform (e.g.,~by replacing higher capacity edges with edge multiplicities). Suppose that the routes are selected and consider a given set of paths $\{ p_{s_i, t_i} \}_i$, one path for each packet. Let $\dil$ denote the length of the longest path among these. Also, let the congestion of each edge $e$ be the number of paths that include it and let the overall $\con$ denote the maximum congestion over all edges $e$. While the best completion time depends on the exact model of how are the packets coordinated (e.g., centralized scheduling of packets vs. distributed scheduling), in these settings the completion time is near-optimally characterized by the quantity $\con+\dil$. Clearly, delivering the packets along these paths requires at least $\max\{\con, \dil\} \geq (\con+\dil)/2$ time units (i.e., the completion time is at least $(\con+\dil)/2$). Leighton, Maggs, and Rao show a centralized scheduling algorithm that would deliver all the packets to their destinations with completion time of $O(\con+\dil)$ time units. Moreover, just using their basic \emph{random delays} idea, we can define a schedule that can be implemented in a distributed setting with a near-optimal completion time of at most $O(\log n)\cdot (\con+\dil)$. Thanks to this simple \emph{random delays} idea, if we ignore logarithmic factors, the routing question boils down to the route selection problem while minimizing $\con+\dil$, which is the question we focus on the in remainder of the paper.

\paragraph{Oblivious route selection:} In the oblivious case, we are given a set of sources and destinations $\calD = \{(s_i, t_i)\}_i$, which we call the \textbf{demand}. If the demand is known in advance (i.e., all $\{s_i\}_i$ and $\{t_i\}_i$ are known), the route selection problem can be solved in polynomial time, giving a set of routes that has $\con+\dil$ within a constant factor of the optimum, by a classic result of Srinivasan and Teo~\cite{srinivasan2001constant}. Put together with the aforementioned scheduling result of Leighton et al.~\cite{leighton1994packet}, this gives a constant approximation algorithm for the completion time in packet routing. However, in this scheme, the routes selected by different requests heavily depend on each other, and devising these routes requires central control of the entire network. A much more common scenario in networking is that the packet delivery requests arrive at various points in the network. It is much more desirable if one can select the route of each packet just based on its source and destination, and in a manner \emph{oblivious} to all the other packet routing requests. More formally, a (probabilistic) \emph{routing scheme} can be summarized as follows.

\begin{definition}\label{def:oblivious-routing}
  A \textbf{routing scheme} $R$ for an undirected graph $G=(V,E)$ is a collection of $|V|^2$ distributions $R = \{R_{u, v}\}_{u, v \in V}$, where for each pair of nodes $u, v \in V$, we have one distribution $R_{u, v}$ over paths between $u$ and $v$.
\end{definition} 

A routing scheme can be used to \emph{obliviously} route requests $\{(s_i, t_i)\}$ in the following straightforward way: Given a routing scheme $R$, the $i^{th}$ path is independently sampled from $R_{s_i,t_i}$. Note that each request is routed independently of (the existence of) other requests, hence the routing is \emph{oblivious}.


\paragraph{Quality measures:} Our goal is to find routing schemes which, for every demand, guarantee that the obliviously selected paths are competitive with the optimal (demand-dependent) set of paths in terms of some quality measure. We can measure the quality of the selected paths
using various functions, including the maximum or average congestion,
$\ell_p$ norm of edge congestion, the maximum dilation, etc. Given the discussions above, our primary measure of interest will be the summation $\con+\dil$. As noted before, thanks to the random delays technique for scheduling~\cite{leighton1994packet}, a routing scheme that has polylogarithmic competitiveness in terms of the $\con+\dil$ measure provides a routing scheme that has polylogarithmic competitiveness in terms of the completion time to deliver all packets. 

\subsection{Prior work on oblivious routing}
We next discuss the prior work on routing schemes that are obliviously competitive for other measures, and some of the known obstacles towards being competitive in $\con+\dil$.
 
\paragraph{Results on special graphs:} Valiant and Brebner~\cite{valiant1981universal} were the first to study oblivious packet routing. They focused on the case where the network is a hypercube and showed that any permutation can be routed with completion time $O(\log n)$. Their path selection is based on the ``Valiant's trick'' of routing from the source $s_i$ to a random node $q$ and from there to the destination $t_i$, where the paths from $s_i$ to $q$ and from $q$ to $t_i$ are greedy (fixing differing dimensions one by one). Following them, there have been a number of routing schemes that are obliviously competitive in terms of $\con+\dil$ in a range of special graphs, including expanders, Caley graphs, fat trees, meshes, etc.~\cite{rabin1989efficient, upfal1984efficient, busch2005oblivious, scheideler2006universal, busch2008optimal,  busch2010optimal} See the thesis of Scheideler~\cite{scheideler2006universal} and the survey of R\"acke~\cite{racke2009survey} for more on related work.

\paragraph{Congestion-competitive oblivious routing:} A prominent highlight in prior work is \emph{congestion-competitive} oblivious routing, a topic which was initiated by R\"{a}cke's seminal paper~\cite{racke2002minimizing}, and through a beautiful line of work \cite{bienkowski2003practical, harrelson2003polynomial, azar2004optimal}, culminated in the following celebrated result of R\"{a}cke's~\cite{racke2008optimal}: For every undirected graph, there is a polynomial-time algorithm to build a routing scheme such that for every demand $\calD$, the routing obliviously produces a collection of routes which are $O(\log n)$-competitive in terms of $\con$, compared to the optimal collection of (demand-dependent) routes for $\calD$.

\begin{theorem}[R\"{a}cke~\cite{racke2008optimal}]
For every undirected (multi)-graph $G=(V, E)$, there exists a routing scheme $R = R(G)$ such that for every demand $\calD = \{(s_i, t_i)\}_{i=1}^k$, the maximum expected congestion of routing the demand $\calD$ using $R$ is at most a $O(\log n)$-factor larger than the optimal congestion of $\calD$ in $G$.
\end{theorem} 

R\"{a}cke's routing scheme has the additional property of being \emph{tree-based} (see \Cref{sec:tree-like-routing-impossible} for a formal definition). Moreover, the current state-of-the-art approaches are all based on tree-based routing schemes, which greatly simplify the process of constructing competitive oblivious routings. A significant challenge this paper needed to overcome is the fact that tree-based routing schemes do not exist in our setting of jointly minimizing the congestion and dilation. To this end, we develop a theory of constructing routings using \emph{partial trees} which allow for greater flexibility at the cost of increased intricacy of the construction.


\paragraph{Oblivious routing for congestion and dilation:} Considering that both $\con$ and $\dil$ impact packet delivery, it would be very desirable to be competitive in both, or just their summation. Unfortunately, it is well-known that R\"{a}cke's routing scheme is not competitive in terms of $\dil$ and it can select paths that are arbitrarily longer than the paths in the optimal collection. Because of this, while the routing is competitive in the congestion measure, it is not competitive in the completion time measure, or more concretely in terms of $\con+\dil$. Of course, if we focus on $\dil$ alone, it is trivial to be competitive by simply routing each packet along the shortest path between its source and destination. 

\paragraph{Oblivious routing for other measures:}
Gupta, Hajiaghayi, and R\"acke~\cite{englert2009oblivious} study a range of rather general measures for oblivious routing. Suppose for each commodity $i\in\{1, ..., k\}$, we want a (fractional) flow from source $s_i$ to destination $t_i$. For a given set of flows for the commodities, for each edge $e$, let us use $f_i(e)$ to denote the amount of commodity-$i$ flow passed through $e$. Consider a \emph{load function} $L_e:\mathbb{R}^k_{+}\rightarrow \mathbb{R}_{+}$ where $L_e := \ell(f_{1}(e), \dots , f_{k}(e))$ defines the \emph{load} of edge $e$. Gupta et al.~\cite{gupta2006oblivious} provide two results: (1) an $O(\log^2 n)$ competitive oblivious routing algorithm for the summation of loads of different edges, assuming that the load function $\ell$ is the class of monotone sub-additive functions, (2) an $O(\log^2 n \log\log n)$ competitive oblivious algorithm for the maximum of the loads of different edges assuming that the load function is a norm. In either case, their oblivious routing does not need to know the load function $\ell$. 
Englert and R\"{a}cke~\cite{englert2009oblivious} extended this framework and presented an $O(\log n)$-competitive oblivious algorithm for the case that the load function is a monotone norm, and we take the $\ell_p$ norm of the loads of different edges as our measure for competitiveness.


\paragraph{An impossibility?} The goal of being competitive in both $\con$ and $\dil$ has been discussed in the literature of oblivious routing as an ultimate goal~\cite{aspnes2006eight, racke2009survey}. However, the discussion often soon concludes in an ``impossibility'': Suppose that we interpret the goal as being competitive in terms of $\con$ and also in terms of $\dil$, simultaneously. This is not possible. Consider two neighboring nodes $u$ and $v$ that, besides the edge between them, are also connected with $\Theta(\sqrt{n})$ disjoint paths of length $\Theta(\sqrt{n})$. If there are $\sqrt{n}$ packets that should go from $u$ to $v$, to be $\poly(\log n)$-competitive in $\con$, at least $\sqrt{n} - \poly(\log n)$ of packets should be routed through those long paths. In an obliviously-competitive routing scheme, that means each packet should be routed through that long path with probability $1-o(1)$. But then, if we consider a demand with just one packet, the optimal $\dil$ is simply $1$, while this routing scheme will have  $\dil=\Theta(\sqrt{n})$ with probability $1-o(1)$. We note that this example also shows that, even if we make the summation $\con+\dil$ as our measure for competitiveness, $\poly(\log n)$-competitive oblivious routing appears impossible. Indeed, noting this apparent impossibility for general graphs, Aspnes et al.~\cite{aspnes2006eight} asked for a workaround in their list of prominent open problems in distributed computing. Their suggestion was that it might be still feasible for special networks: ``\emph{Another important open problem is to find classes of networks in which oblivious routing gives C+D [i.e., $\con+\dil$] close to the off-line optimal... Such a result have immediate consequences in packet scheduling algorithms.}''

Our aim in this paper is to have a solution for \emph{all} graphs. This calls for re-examining the above impossibility argument. The astute reader notices that, in this simple example, changing the requirement slightly makes the problem possible and still perfectly useful: if we are given a target $\dil$ upper bound $h$, which is within $\poly(\log n)$ factor of the optimal $\dil$ for the input instance, then the routing scheme can select the routes so as to remain below this target $h$ and still be  $\poly(\log n)$-competitive in terms of the $\con$.  Similarly, when $\con+\dil$ is the measure, given a $\poly(\log n)$ factor upper bound on the value of the optimum, we can get a $\poly(\log n)$-competitive oblivious routing in terms of $\con+\dil$. Of course, in these definitions, the routes may depend on the given $\dil$ upper bound, or the upper bound on $\con+\dil$. The assumption of having these upper bounds can be removed by standard ideas such as guessing and doubling. Furthermore, this choice also has the added flexibility of exploring the optimal trade-off between feasible values for $\con$ and $\dil$. We note that even though in this simple example the problem becomes possible and easy, achieving such a result for general graphs is far from trivial. Indeed, since we can set the upper bound arbitrarily high, this is a strictly stronger requirement than requiring competitiveness in $\con$ alone (as in results of \cite{racke2002minimizing,racke2008optimal} mentioned above). This is exactly the goal that we achieve in this paper.

\subsection{Our contribution}
We show a routing scheme that, given a $\dil$ bound $h$, is obliviously $\poly(\log n)$-competitive in terms of $\con$ compared to the best $\con$ achievable via paths of length $h$, and our routing scheme uses paths of length at most $h \cdot \poly(\log n)$.

\begin{theorem}\label[theorem]{thm:main-informal}
(\textbf{Informal})
For every undirected (multi)-graph $G = (V, E)$ and every dilation bound $h \ge 1$, there exists a routing scheme $R = R(G, h)$ such that for every demand $\calD = \{(s_i, t_i)\}_{i=1}^k$, the selected paths have a hop-length of at most $h \cdot O(\log^7 n)$, and in which the expected maximum $\con$ is within an $O(\log^2 n \cdot \log^2 (h \log n))$ factor of the optimum $\con$ when routing along paths of length at most $h$. 
\end{theorem}

See \Cref{thmGeneralRouting} for the formal statement. \Cref{thm:main-informal} directly gives us a way of obtaining a routing scheme that is obliviously competitive in terms of the completion time. As explained above, thanks to known scheduling results, we can simply focus on the summation $\con+\dil$. If we set $h$ to be the optimal value of $\con+\dil$, or a constant-factor approximation of it, the produced routing is $\poly(\log n)$-competitive in $\con+\dil$, and thus also  $\poly(\log n)$-competitive in terms of the completion time, via random delays~\cite{leighton1994packet}.



\subsection{Follow-up work} \rev{The computational aspects of hop-constrained oblivious routings have seen significant amount of follow-up work to the conference version of this paper. Most notably, very recent work by Ghaffari, R\"acke, and Ghaffari~\cite{hopexpander2022} has shown that $h$-hop oblivious routing distributions can be constructed in $\poly(h) \cdot m^{1+o(1)}$, which is almost linear for the important case of $h \le n^{o(1)}$.}


\section{Preliminaries}\label{sec:hbor-prelims}
We first give some common notation we use throughout.

\paragraph{General:}
We denote by $[k] = \{1, 2, \ldots, k\}$ for some non-negative integer $k$ and $A \sqcup B$ denotes the disjoint union of $A$ and $B$. We often use the Iverson bracket notation $\1{\text{condition}}$ which evaluates to $1$ when the $\mathrm{condition}$ is true and $0$ otherwise.
We assume that all graphs are undirected and we typically assume the existence of an underlying graph $G = (V, E)$ with $n := |V|$.

\paragraph{Weighted graphs:}
A \textbf{weighted graph} $G = (V, E, w)$ is an undirected graph $(V, E)$ along with \textbf{weights} $w : E \to [1/L, L]$. Equivalently, a \textbf{capacitated graph} $G = (V, E, c)$ is an undirected graph along with \textbf{capacities} $c : E \to [1/L, L]$. Here, the value $L$ is the \textbf{aspect ratio} and throughout this paper we assume that $L = \poly(n)$. Many intermediate results only apply to complete graphs, i.e., where $E = \binom{V}{2}$, hence we will explicitly disambiguate between general and complete weighted graphs. In the context of this paper, assuming graphs are complete is mostly without loss of generality. In particular, one can often transform any non-complete weighted graph $G = (V, E, w)$ with aspect ratio $L$ into an equivalent complete graph $G'$ with aspect ratio $L' = n^2 \cdot L$ which gives a weight of $L'$ (or $1/L'$ in maximization problems) to any edge not in $E$ without affecting the results. We use both $w_e$ and $w(e)$ to refer to the weights \rev{and we sometimes disambiguate by specifying the graph in the subscript (e.g., $w_G, C_G$). We also sometimes use the vector notation (e.g., $\vec{w}, \vec{c}$) to emphasize that $c$ is a vector and not a scalar.}

\paragraph{Distances and hop-constraints:} Let $p = (p_0, p_1, \ldots, p_\ell)$ be a path in $G = (V, E, w)$. We denote the number of \textbf{hops} in $p$ with $\hop(p) := \ell$ and the sum of weights of the edges in the path with $w(p)$. Paths are assumed to be non-simple unless explicitly stated otherwise. We define the distance between $u, v \in V$ as $\dist(u, v) := \min \{ w(p) \mid \text{path $p$ between } u, v \}$. Furthermore, we define the \textbf{hop-constrained distance} with $h > 0$ as $\dist\at{h}(u, v) := \min \{ w(p) \mid \text{path $p$ between } u, v \text{ with } \hop(w) \le k \}$. \rev{If there are no paths of at most $h$ hops between $u, v$ we define $\dist\at{h}(u, v) := \infty$.}

\paragraph{Trees:} Trees can be either unweighted $T = (V, E)$ or weighted $T = (V, E, w_T)$. Trees are often rooted, in the sense that there is a special node $r \in V$. With $T_{u,v}$ we denote the (unique) path on $T$ between $u$ and $v$. For example, in a weighted tree $T$ it holds that $d_T(u, v) = w_T(T_{u, v})$ for all $u, v \in V$.

\subsection{Approximating hop-constrained distances}\label{sec:approx-hop-constrained-distances}

In this section, we describe an important ingredient from prior work: partial tree embeddings, and how to use them to approximate hop-constrained distances. To give some context, the seminal results of \cite{bartal1996probabilistic,bartal1998approximating,fakcharoenphol2003tight} have shown that any metric space can be approximately embedded into a (distribution over) weighted trees. This has led to major advances in the field of approximation algorithms because many NP-hard problems on general graphs can be exactly solved on trees in polynomial time (see, for instance, the survey \cite{gupta2011approximation}).

However, such tree embedding results are not applicable to problems involving hop-constrained distances $d\at{h}_G$ since they do not form a metric space. To address this issue, very recent work~\cite{haeupler2020tree} proposed using \emph{partial} tree embeddings, where only a fraction of nodes are embedded in any single tree. Their paper also shows that using standard (non-partial) trees in such a setting necessarily leads to unsatisfactory polynomial losses in approximation. We formalize the notion of embedding a graph into a tree.
\begin{definition}\label{def:PartialTreeEmbedding}
  A \textbf{partial tree embedding} $(T, T^G)$ on a graph $G = (V(G),\allowbreak E(G))$ consists of a rooted tree $T = (V(T), E(T))$ with $V(T) \subseteq V(G)$, and a mapping $T^G$ which maps \rev{each tree edge $\{u, v\} = e \in E(T)$ to a path $T_e^{G}$ in $G$ with endpoints $u$ and $v$}.
\end{definition}
We extend the notation from \Cref{def:PartialTreeEmbedding} to nodes in $T$ which are not adjacent: For any two vertices $u, v \in V(T)$, if $e_i$ is the $i^{th}$ edge in $T_{u,v}$ (ordered from $u$ to $v$) then $T_{u,v}^G := T_{e_1}^G \oplus T_{e_2}^G \oplus \ldots$ where $\oplus$ is concatenation. Finally, we extend the notion to weighted graphs $G$ where the distances in the embedding must ``dominate'' the distances in $G$.

\begin{definition}\label{def:DominatingEmbedding}
  A partial tree embedding $(T, T^G)$ is \textbf{dominating} if (i) it is defined with respect to \emph{weighted} graph $G = (V(G), E(G), w_G)$, (ii) $T = (V(T), E(T), w_T)$ is a weighted tree, and (iii) $w_G(T^G_{u, v}) \le d_T(u, v)$ for all $\{u, v\} \in E(T)$.
\end{definition}
Note that if $(T, T^G)$ is a dominating partial tree embedding, we have that $w_G(T^G_{u, v}) \le d_T(u, v)$ for all $u, v \in V(T)$.

Distributions over partial tree embedding are sufficiently expressive to approximate hop-constrained distances $d\at{h}_G$ in any graph $G$. However, the approximation here is bicriteria: the distances are stretched by a factor $\alpha \ge 1$ and the hop lengths are stretched by a factor $\beta \ge 1$. Furthermore, while some nodes are missing from any particular tree, each node $v \in V(G)$ must be embedded in at least an $1 - \eps$ fraction of the trees.

\begin{definition}[$h$-Hop Partial Embedding Distribution]\label{def:partialEmbeddingDistrib}
  An \textbf{$h$-hop partial embedding distribution} is a distribution over dominating partial tree embeddings $\calT$ on a weighted graph $G$. We annotate $\calT$ with the following properties:
  \begin{itemize}  
  \item $\calT$ has \textbf{hop stretch} $\beta \ge 1$ if each partial tree embedding $(T, T^G) \in \supp(\calT)$ has $\hop(T^G_{u, v}) \le \beta \cdot h$ for all $u, v \in V(T)$.
  \item $\calT$ has \textbf{exclusion probability $\eps > 0$} if each node $v \in V(G)$ we have $\Pr_{\rev{(T, \cdot) \sim \calT}}[\allowbreak v \in V(T)] \ge 1 - \eps$.
  \item $\calT$ has \textbf{expected distance stretch} $\alpha \ge 1$ if $\E_{(T, \cdot) \sim \calT} \left[ d_T(u,v) \cdot \1{ u,v \in V(T) } \right ] \allowbreak \le \alpha \cdot d_G\at{h}(u, v)$ for all $u, v \in V(G)$.    
  \end{itemize}
\end{definition}

\begin{theorem}[\cite{haeupler2020tree}]
  \label{thm:hst}
  For every (complete) weighted graph $G$ with polynomially-bounded weights, every $h \ge 1$, and every $0 < \eps < 1/3$ there exists an $h$-hop partial embedding distribution with hop stretch $O(\frac{\log^3 n}{\eps})$, expected distance stretch $O(\log n \cdot \log \frac{\log n}{\eps} )$, and exclusion probability $\eps$. Moreover, the distribution can be sampled in polynomial time.
\end{theorem}

We elaborate on the notion in which $d\at{h}$ is approximated by $h$-hop partial embedding distributions. An alternative way to view the expected distance stretch $\alpha$ is to interpret it as a variant of the expected stretch of the conditional distribution, i.e, for $\eps < 1/3$ we have
\begin{align*}
  \E_{(T, \cdot) \sim \calT}\left[d_T(u,v) \cdot \1{u, v \in V(T)} \right] = \Theta\left(\E_{(T, \cdot) \sim \calT}\left[d_T(u,v) \mid u, v \in V(T) \right]\right) .
\end{align*}
With this in mind and the fact that partial embedding distributions $\calT$ are over dominating embeddings, we can rewrite the guarantee of \Cref{def:partialEmbeddingDistrib} in the following way for every $u, v \in V(G)$:
\begin{align*}
  d_G\at{\beta \cdot h}(u, v) \le \E_{(T, \cdot) \sim \calT}\left[d_T(u,v) \mid u, v \in V(T) \right] \le O(\alpha) \cdot d_G\at{h}(u, v) .
\end{align*}

\subsection{Hop-constrained oblivious routings}\label{sec:hop-constrained-definitions}

In this section we give a formal definition of hop-constrained oblivious routings.

\medskip

\textbf{Fractional demands and routings.} The discussion in the introduction mostly focused on \emph{integral} demands and routings---where the demand was a sequence of $(s_i, t_i)$ pairs and the routes were single paths between $s_i$ and $t_i$. However, we express our technical results in the more general---fractional---setting (see further discussion about this later in this section). To this end, a set of requests $\{( s_i, t_i )\}_i$ are abstracted and generalized via the (fractional) demand matrix $\calD : V \times V \to \mathbb{R}_{\ge 0}$ where $\calD_{s, t}$ intuitively represents the number of requests between $s$ and $t$. In other words, $\{(s_i, t_i)\}_i$ would correspond to the demand (matrix) $\calD_{u, v} := |\{ i : s_i = u, t_i = v \}|$. Similarly, a path between $s$ and $t$ will typically be replaced by a distribution of paths between $s$ and $t$. This requires us to slightly revisit the notions of dilation and congestion, which were previously defined only in the integral case. The dilation of a distribution of paths is $h$ if the distribution is supported over paths of at most $h$ hops. Next, we formally introduce the notion of flows, capacitated graphs and congestion.
\begin{definition}[Flows]
  A \textbf{flow} $f$ is a non-negative vector indexed over the edges of the underlying graph $G$, i.e., $f \in \mathbb{R}_{\ge 0}^{E(G)}$. Each (possibly non-simple) path $p$ has a naturally associated flow which we denote as $\flow(p) = ( \flow(p, e) )_{e \in E(G)} \allowbreak \in \mathbb{R}_{\ge 0}^{E(G)}$ where $\flow(p)_e$ is defined as the number of times $p$ goes through $e$.
\end{definition}
\rev{We establish a partial order on the flow vectors $f, g \in \mathbb{R}^{E(G)}$ where we write $f \le g$ when the inequality holds element-wise.}

We express our results on capacitated graphs, which generalize multi-graphs, in the sense that a multi-graph with $k_e$ copies of an edge $e$ corresponds to a capacitated graph with an edge capacity $c_e := k_e$. This allows us to define the congestion of a flow with respect to capacities.
\begin{definition}[Capacitated graphs and congestion]
  Capacitated graphs $G = (V, E, c)$ are undirected graphs with \textbf{capacities} $c : E \to [1/L, L]$, where $L = \poly(n)$ is the aspect ratio. Given a flow $f \in \mathbb{R}_{\ge 0}^{E}$, we define the \textbf{congestion} $\cong(f) = \max_{e \in E} f_e / c_e$ to be the maximum ratio of flow over capacity, across all edges.
\end{definition}

\medskip

\textbf{Optimally routing a demand via a (hop-constrained) routing scheme.} Given a particular demand $\calD$, we are interested in the best hop-constrained (fractional) routing $R = \{ R_{s, t} \}_{s, t \in V}$ of $\calD$. This is formalized in the following definition.%
\begin{definition}[Optimal hop-constrained routings]
  A \textbf{demand} is a matrix $\calD : V \times V \to \mathbb{R}_{\ge 0}$. Given a demand $\calD$ and a routing scheme $\{R_{s, t}\}_{s, t \in V}$ for a capacitated graph $G$, \rev{we define the \textbf{flow routing the demand $\calD$ using $R$} as
  \begin{align*}
    \flow(\calD, R) = \sum_{s, t} \E_{p \sim R_{s,t}}[ \flow(p) \cdot \calD_{s, t} ] \in \mathbb{R}^{E(G)}_{\ge 0} .
  \end{align*}
The \textbf{congestion of routing $\calD$ using $R$} is $\cong(\calD, R) := \cong(\flow(\calD, R))$.}
Finally, given a demand $\calD$ and a hop constraint $h \ge 1$, we define the \textbf{optimal $h$-hop routing of $\calD$}, denoted by $\opt\at{h}(\calD)$, as the minimum congestion $\cong(\calD, R)$ over all possible routing schemes $R = \{ R_{s, t} \}_{s, t \in V}$ supported over paths of at most $h$ hops.
\end{definition}

\rev{As an intuitive explanation for the definition of $\flow(\cdot, \cdot) \in \mathbb{R}^{E(G)}_{\ge 0}$, we note that $\flow(\calD, R) \in \mathbb{R}^{E(G)}_{\ge 0}$ is a vector where the coordinate corresponding to $e \in E(G)$ is exactly the expected number of routes going over an edge $e$ when routing the demand $\calD$ using the routing scheme $R$ in an oblivious way.}

We reiterate that all routings and distributions in our paper are fractional, in the sense that a unit demand from $s$ to $t$ is carried over a distribution of paths, incurring fractional congestion on each one of these paths. Another important point of emphasis in the definition of $\opt\at{h}(\calD)$ is that the routing scheme $P$ can depend on the demand $\calD$, i.e., is \emph{adaptive} to the demand. In contrast, a good oblivious routing is a single routing scheme that is \emph{obliviously competitive} with respect to \emph{all demands} $\calD$. We now define the principal concept in our paper.
\begin{definition}
  An \textbf{$h$-hop oblivious routing} for a graph $G = (V, E)$ is a routing scheme $R$ that is additionally annotated in the following way:
  \begin{enumerate}
  \item $R$ has \textbf{hop stretch} $\beta \ge 1$ if for all $s, t \in V$ all paths $p \in \supp(R_{s, t})$ have $\hop(p) \le \beta \cdot h$.
  \item $R$ has \textbf{congestion approximation} $\alpha \ge 1$ if for all demands $\calD : V \times V \to \mathbb{R}_{\ge 0}$ we have that $\cong_G(\calD, R) \le \alpha \cdot \opt\at{h}(\calD)$.
  \end{enumerate}
\end{definition}

\textbf{Integral vs. fractional routings.} Our choice to express our results in the fractional setting has multiple benefits. For one, the demand in our (fractional) setting is scale-invariant, in that a routing scheme that is competitive with respect to $\calD$ will be competitive with respect to $\gamma \cdot \calD$, for any $\gamma > 0$. Furthermore, one can easily recover the integral setting from the fractional one, making our choice more general. We elaborate on this. Suppose we are given an $h$-hop oblivious routing $R$ with hop stretch $\beta$ and congestion approximation $\alpha \gg \log n$. Given a set of requests $\{(s_i, t_i)\}_{i}$, suppose that some set of paths $\{p^*_i\}_{i}$ with at most $h$ hops connecting the source-sink pairs has optimal congestion $C^*$. For each request $i$ we independently randomly sample a path $p'_i \sim R_{s_i, t_i}$.

We now argue that $\{ p'_i \}$ has dilation $\beta \cdot h$ and congestion at most $O(\alpha) \cdot C^*$ with high probability. The dilation bound follows from definition. We now argue that for each edge $e$ the expected number of times $\{p'_i\}$ crosses $e$ is $\alpha \cdot C^* \cdot c_e$. First, for all $u, v$ we set $\calD_{u, v} := |\{ i : s_i = u, t_i = v \}|$ and then let $P_{u, v}$ be a uniform distribution over the $\calD_{u, v}$ paths with endpoints $u, v$. We note that $\opt\at{h}(\calD) \le \cong_G(\calD, P) \le C^*$, where the first inequality is by definition and second is due to $\cong_G(\calD, P)$ being exactly equal to the congestion of a set of paths $\{p^*_i\}_i$. Therefore, the expected number of drawn paths $\{p'_i\}$ crossing an edge $e$ is at most $\alpha \cdot c_e \cdot \opt\at{h}(\calD) \le \alpha \cdot c_e \cdot C^*$, as required.

We now argue that the congestion of $\{p'_i\}_i$ is at most $\alpha \cdot C^*$ with high probability. Since each path was drawn independently at random and the paths can always assumed to be simple (simplying a path does not increase the congestion or the dilation), the number of paths $\{p'_i\}$ crossing an edge $e$ can be seen as a sum of independent $\{0, 1\}$-variables. We can apply a standard Chernoff bound and conclude that this number is at most $O(c_e \cdot \alpha \cdot C^* + \log n)$ with high probability. Since $C^* \cdot c_e \ge 1$ and $\alpha \gg \log n$, we have that $O(c_e \cdot \alpha \cdot C^* + \log n) = O(c_e \cdot \alpha \cdot C^*)$ with high probability. Union bounding over all edges, we conclude that the congestion of $\{ p'_i \}_i$ is at most $O(\alpha) \cdot C^*$ with high probability, as required.

\section{Hop-Constrained Oblivious Routing: A Technical Overview}

We now formally state our main result.

\begin{restatable}{theorem}{thmGeneralRouting}\label{thmGeneralRouting}
  For every (general) capacitated graph $G = (V, E, c)$ \rev{with polynomially\hyp{}bounded capacities} and every $h \ge 1$, there exists an $h$-hop oblivious routing with hop stretch $O(\log^7 n)$ and congestion approximation $O(\log^2 n \cdot \log^2 (h \log n))$.
\end{restatable}

\textbf{Remark.} Notice that there is a small discrepancy between the definition of congestion in \Cref{thmGeneralRouting} and its informal counterpart \Cref{thm:main-informal}. The former talks about fractional routings (e.g., maximum expected congestion), while the later is about integral routings (e.g., expected maximum congestion). However, our formal statement (\Cref{thmGeneralRouting}) implies the informal one, as argued in \Cref{sec:hop-constrained-definitions}.

The rest of the paper is structured as follows. In \Cref{sec:tree-like-routing-impossible} we explain why previous (tree-based) approaches fail to attain oblivious routings with hop constraints. In \Cref{sec:how-does-the-distribution-look-like} we give an overview of the routing scheme that achieves the guarantees of the hop-constrained oblivious routing. Finally, in \Cref{sec:lifting-d1-to-routings} we prove the guarantees of our routing scheme.


\subsection{Tree-based hop-constrained oblivious routings cannot have good guarantees}\label{sec:tree-like-routing-impossible}
   
In this section, we showcase a barrier that prevented prior approaches from achieving $\poly(\log n)$-competitive hop-constrained oblivious routing. State-of-the-art oblivious routings are generally \emph{tree-based routings}\footnote{We also note that routings supported on hierarchically separated trees or HSTs can be converted into tree-based routings with at most a constant loss in their guarantees.}~\cite{racke2008optimal,racke2009survey,bienkowski2003practical,englert2009oblivious,williamson2011design}, i.e., where a demand from $s$ to $t$ is routed by randomly sampling a tree $T$ from a fixed distribution $\calT$ and then picking the tree-defined path $T^G_{u, v}$. We show that tree-based hop-constrained oblivious routings cannot have good guarantees.
\begin{definition}
  A \textbf{complete tree embedding} $(T, T^G)$ of a graph $G = (V(G), \allowbreak E(G))$ consists of a tree $T = (V(T), E(T))$ where $V(G) = V(T)$, and a mapping $T^G$ which maps every edge $e \in E(G)$ to a path in $G$ between $e$'s endpoints.
\end{definition}
\noindent We extend the definition of $T^G_{u, v}$ when $u, v$ are not adjacent in the natural way as in \Cref{def:PartialTreeEmbedding}. Note that the tree from the complete tree embedding is not necessarily a subtree of $G$, but it does contain \emph{all} nodes of $G$.

\begin{definition}
  A \textbf{tree-based routing scheme} $R^\calT$ is a routing scheme that is induced by a distribution over complete tree embeddings $\calT$ in the following way: we sample from $R^\calT_{s, t}$ by sampling an embedding $(T, T^G) \sim \calT$ and returning $T^G_{s, t}$.
\end{definition}


\paragraph{The $\calD\at{1}$ demand.} For a capacitated graph $G = (V, E, c)$ we define a special demand $\calD\at{1}$ which has a unit demand across every edge $e$ (for each unit of capacity), i.e., $\calD\at{1}_{s, t} := \1{ \{s, t\} \in E } \cdot c_{\{s, t\}}$. This demand is particularly important for both the congestion-only and hop-constrained oblivious routings, as we shortly explain. Tree-based routings are especially suitable for constructing oblivious routings: if a tree-based routing incurs congestion $\alpha$ on the $\calD\at{1}$ demand, then it is $\alpha$-congestion-competitive with respect to all possible demands. This greatly simplifies the design of good tree-based routing schemes: one only needs to ensure the single $\calD\at{1}$ is routed in a good manner. The following statement formalizes the claim in the tree-based routing case (note that the discussion up to this point is for the unconstrained-hop setting, hence $h = \infty$).

\begin{restatable}{lemma}{lemmaDToAllTrees}\label{lemma:d1-to-all-trees}
  Let $G = (V, E, c)$ be a capacitated graph and suppose that a tree-based routing scheme $R^\calT = \{R^\calT_{s, t}\}_{s, t \in V}$ achieves $\cong_G(\calD\at{1}, R^\calT) \le \alpha$, where $\calD\at{1} : V \times V \to \mathbb{R}_{\ge 0}$ with $\calD\at{1}_{s, t} := \1{ \{s, t\} \in E } \cdot c_{\{s, t\}}$. Then for every demand $\calD : V \times V \to \mathbb{R}_{\ge 0}$ we have $\cong_G( \calD, R^\calT ) \le \alpha \cdot \opt\at{\infty}(\calD)$.
\end{restatable}
\begin{proof}
  The claim is implicit in, e.g., \cite{racke2008optimal}. The full proof is recreated for completeness in \Cref{sec:proof-lemma-d1-to-all-trees}. 
\end{proof}

Unfortunately, hop-constrained oblivious routings with polylogarithmic hop stretch and congestion approximation cannot come from tree-based routings. This observation prevents all prior work for general graphs the authors are aware of from achieving good-quality routings that control both the congestion and dilation.

\begin{lemma}
  There exists an infinite family of graphs with unit capacities and diameter $4$ such that for every graph $G$ in the family the following holds. For every $h \ge 1$, any tree-based $h$-hop oblivious routing $R^\calT$ for $G$ with hop stretch $\beta$ and congestion approximation $\alpha$ has $\alpha \cdot \beta \cdot h \ge \Omega(\sqrt{n})$.
\end{lemma}
\begin{proof}
  By definition of hop stretch, we have that $R^\calT$ is supported over paths of length at most $\beta h$.
  Furthermore, due to the congestion approximation being at most $\alpha$, the congestion of $R^\calT$ on the demand $\calD\at{1}$ is $\alpha \cdot \opt\at{h}(\calD\at{1}) = \alpha$. Since $R^\calT$ is a tree-based routing, we conclude via \Cref{lemma:d1-to-all-trees} that for all demands $\calD$ the congestion of $R^\calT$ on $\calD$ is $\alpha$-competitive with $\opt\at{\infty}(\calD)$: $\cong_G( \calD, R^\calT ) \le \alpha \cdot \opt\at{\infty}(\calD)$.
  In other words, $R^\calT$ is also a $\infty$-hop oblivious routing with congestion approximation $\alpha$ (in spite of being supported only on paths of length $\beta h$).

  We now construct a graph $G$ with $n+1$ vertices which exhibits our bound. We take $\sqrt{n}$ paths $p_1, p_2, \ldots, p_{\sqrt{n}}$ of length $\hop(p_i) = \sqrt{n} - 1$. Label the first and last node of $p_i$ with $s_i$ and $t_i$. For each path $p_i$ we connect each $j^{th}$ node to the $j^{th}$ node of $p_1$. Finally, we create a new node $r$ and connect it to all nodes on $p_1$ with edges we call ``uplinks''. The diameter of $G$ is clearly $4$ (e.g., the hop distance between any node and $r$ is $2$). We consider the demand $\calD_{s, t} = \1{\exists i \text{ such that } (s,t) = (s_i, t_i)}$. Clearly, $\opt\at{\infty}(\calD) \le 1$ since the demand between $s_i$ and $t_i$ can be sent across the $(\sqrt{n} - 1)$-hop path $p_i$, resulting in congestion $1$. Furthermore, any path between the start and end of some $p_i$ of hop length at most $\beta h$ must cross one of the first $\beta h$ uplinks. Since there are $\sqrt{n}$ such demands, we conclude that the congestion of at least one of the uplinks is $\frac{\sqrt{n}}{\beta \cdot h} \le \alpha \cdot \opt\at{h}(\calD) = \alpha$.
\end{proof}
\textbf{Remark.} The $\sqrt{n}$ bound can be improved to $\tilde{\Omega}(n)$ for unit-capacity graphs of diameter $O(\log n)$ using  the well-known worst-case network family from~\cite{dassarma2012distributed}. 

\subsection{An overview of the hop-constrained oblivious routing}\label{sec:how-does-the-distribution-look-like}

In this section, we give an overview of our hop-constrained oblivious routing.

While one cannot obtain hop-constrained oblivious routings via tree-based routings (i.e., complete tree embeddings, as argued in \Cref{sec:tree-like-routing-impossible}), we show that distributions over \emph{partial tree embeddings} yield useful results. We first introduce the notion of \textbf{$\calD\at{1}$-routers}, which are distributions over partial tree embeddings $\calT$ and are a crucial building block of hop-constrained oblivious routings. We describe several important aspects of $\calD\at{1}$-routers before we formally define them in \Cref{def:d1-router}.

\begin{enumerate}
\item First, we explain how $\calT$ induces an ``routing scheme'' $R^\calT = \{ R^\calT_{s, t} \}_{s, t}$. For some $s, t \in V(G)$ we sample a partial tree embedding $(T, T^G) \in \calT$; if $s, t \in V(T)$ we return $T^G_{s, t}$; otherwise, we simply return a special symbol $\bot$ which represents ``failure''. Clearly, $R^\calT$ is not a valid routing scheme in the sense of \Cref{def:oblivious-routing} since $\bot \in \supp(R_{s, t})$, but this is somewhat unavoidable when dealing with partial tree embeddings.

\item $\calD\at{1}$-routers for a graph $G$ get their name from being able to route the $\calD\at{1}$ demand, defined as $\calD\at{1}_{s, t} := \1{ \{s, t\} \in E(G) } \cdot c_{\{s, t\}}$. While this does not directly guarantee good congestion approximation on all demands (unlike tree-based routings, c.f. \Cref{lemma:d1-to-all-trees}), it does lead to certain useful properties. First, we formalize what we mean by routing the $\calD\at{1}$ demand over routing scheme $R^\calT$ which is induced by a distribution over partial tree embeddings $\calT$. For a demand $\calD$ we define $\cong_G(\calD, R^\calT)$ as the maximum expected congestion while only routing non-failures:
  \begin{align*}
    \cong_G(\calD, R^\calT) & = \cong_G\left(\sum_{s, t} \E_{p \sim R^\calT_{s,t}}\left[\ \1{p \neq \bot}\cdot \flow(p) \cdot \calD_{s, t}\ \right] \right) \\
                            & = \cong_G\left(\sum_{s, t} \E_{(T, T^G) \sim \calT}\left[\ \1{s, t \in V(T)}\cdot \flow(T^G_{s, t}) \cdot \calD_{s, t}\ \right] \right) .
  \end{align*}
  Finally, we say that a $\calD\at{1}$-router has \textbf{congestion approximation} $\alpha$ if $\cong_G(\calD\at{1}, R^\calT) \le \alpha$.
  
\item In our hop-constrained setting, we need to control the hop length of the paths over which $\calT$ routes. For this reason, we define the \textbf{dilation} of $\calT$ to be $\beta$ if $\calT$ is supported over partial tree embeddings $(T, T^G)$ where $\hop(T^G_{s, t}) \le \beta$ for all $s, t \in V(T)$.
   
\item When talking about distributions over partial tree embeddings, a new parameter called \textbf{exclusion probability} becomes important. The exclusion probability is $\eps$ if for each node $v \in V(G)$ the probability that $v$ is excluded from the tree $(T, \cdot) \sim \calT$ is at most $\eps$. Note that this parameter also appears when talking about partial tree embeddings that approximate hop-constrained distances (\Cref{sec:approx-hop-constrained-distances}).

\item \rev{It is not immediately clear why $\calD\at{1}$-routers are at all useful. For example, using them directly to route some arbitrary demand $\calD$ does not give a competitive congestion, even if it happens that the entire demand $\calD$ is supported on nodes that appear in all the partial trees in the distribution (i.e., there exists $S \subseteq V(G)$ such that $\supp(\calD) \subseteq S \times S$ and $S \subseteq \bigcup_{(T, \cdot) \sim \calT} V(T)$, where $\calT$ is a $\calD\at{1}$-routing distribution). However, Zuzic's dissertation shows that $\calD\at{1}$-routers are sufficiently powerful to obtain near-optimal congestion+dilation routing when combined with several other ideas which aim to control the adversarial congestion caused by the failed (i.e., $\bot$) routes~\cite[Section 7.6: Routing with Noise]{zuzic2020phd}. In this paper we present an approach that corrects the failed routes in a significantly more general way; the only downside of our approach is that it yields larger polylog factors.} 
\end{enumerate} 

We now give a formal definition equivalent to the above description and state its existence lemma.
\begin{definition}\label{def:d1-router}
  A \textbf{$\calD\at{1}$-router} is a distribution over partial tree embeddings $\calT$ on a capacitated graph $G$ that is additionally annotated in the following way:
  \begin{enumerate}
  \item $\calT$ has \textbf{dilation} $\beta \ge 1$ if each partial tree embedding $(T, T^G) \in \supp(\calT)$ has $\hop(T^G_{u, v}) \le \beta$ for all $u, v \in V(T)$.
  \item $\calT$ has \textbf{exclusion probability $\eps > 0$} if for each node $v \in V(G)$ we have $\Pr[v \in V(T)] \ge 1 - \eps$.
  \item $\calT$ has \textbf{congestion $\alpha \ge 1$} if
    \begin{align*}
      \sum_{\{u, v\} \in E(G)} \E_{(T, T^G) \sim \calT}\left[\ \1{u, v \in V(T)} \cdot \flow(T^G_{u, v}) \cdot c_{\{u, v\}}\ \right] \le \alpha \cdot \vec{c}_G.
    \end{align*}
  \end{enumerate}
\end{definition}

\begin{restatable}{lemma}{lemmaDOneRouter}\label{lemmaDOneRouter} 
  For every (complete) capacitated graph $G$ \rev{with polynomially\hyp{}bounded capacities} and $0 < \eps < 1/3$ there exists a $\calD\at{1}$-router with dilation $O(\frac{\log^3 n}{\eps})$, exclusion probability $\eps$, and congestion $O(\log n \cdot \log \frac{\log n}{\eps})$. 
\end{restatable}
\begin{proof}
  We write a linear program over $\calD\at{1}$-routers with the goal of minimizing the congestion $\alpha$ while satisfying the dilation and exclusion probability properties. Let $(T_1, T_1^G), (T_2, T_2^G), \ldots, (T_Q, T_Q^G)$ be the (finite) set of possible partial tree embeddings of $G$ satisfying the dilation property, i.e., where $\hop(T^G_{u, v}) \le O(\frac{\log^3 n}{\eps})$ for all $u, v \in V(T)$. Given a vector $\lambda \in \{ x \in \mathbb{R}_{\ge 0}^Q \mid \sum_j x_j = 1 \}$ we denote with $\calT(\lambda)$ the distribution over $\{(T_i, T_i^G)\}_i$ where $\Pr[\calT(\lambda) = (T_i, T_i^G)] = \lambda_i$. Furthermore, let $\Lambda$ be the set of vectors $\lambda = (\lambda_i)_i$ satisfying the exclusion probability property, i.e., $\Pr_{(T, \cdot) \sim \calT(\lambda)}[ v \in V(T) ] \ge 1 - \eps$ for all $v \in V(G)$. \rev{Note that $\Lambda$ is a convex polytope.}

  We now present the linear program. Note that $\cong_G(x) \le \alpha$ can be written as $x_e \le \alpha \cdot c_e$ where $c_e$ is the capacity of an edge $e$ in $G$.
  \begin{align*}
    \min_{\alpha, \lambda}. & \quad \alpha \\
    \text{such that} & \quad \lambda \in \Lambda, \text{ and} \\
    \forall e \in E(G) & \quad \sum_{\{u, v\} \in E(G)} c_{\{u, v\}} \cdot \E_{(T, T^G) \sim \calT(\lambda)}\left[ \1{u, v \in V(T)} \cdot \flow(T^G_{u, v})_e \right] \le \alpha \cdot c_e
  \end{align*}

  We dualize the linear program. Note that the primal can be written as $\min \{ \alpha \mid M\lambda \le \alpha, \lambda \in \Lambda \}$ for an appropriately chosen matrix $M$. Therefore, we use the dualization formula
  \begin{align*}
    \min \{ \alpha \mid M\lambda \le \alpha, \lambda \in \Lambda \} = \max_{\ell \ge 0, 1^T \ell = 1} \quad \min_{\lambda \in \Lambda} \quad \ell^T M \lambda .
  \end{align*}
  Using these values, we rewrite the right-hand side of the equation.
  \begin{align*}
    \ell^T M \lambda & = \sum_{e \in E(G)} \ell_e \cdot \E_{(T, T^G) \sim \calT(\lambda)}\left[ \sum_{\{u, v\} \in E(G)} c_{\{u, v\}} \cdot \1{u, v \in V(T)} \cdot \flow(T^G_{u, v})_e \right] / c_e \\
                     & = \sum_{\{u, v\} \in E(G)} c_{\{u, v\}} \cdot \E_{(T, T^G) \sim \calT(\lambda)}\left[ \1{u, v \in V(T)} \sum_{e \in E(G)} \ell_e / c_e  \cdot \flow(T^G_{u, v})_e \right] \\
                     & = \sum_{\{u, v\} \in E(G)} c_{\{u, v\}} \cdot \E_{(T, T^G) \sim \calT(\lambda)}\left[ \1{u, v \in V(T)} \cdot w_{G'(\ell)}(T^G_{u, v}) \right]
  \end{align*}
  In the last line, we introduced a new weighted graph $G'(\ell)$ that is defined as having the same node set and edge set as $G$, while its weights $w_{G'(\ell)}$ are set to $w_{G'(\ell)}(e) = \ell_e / c_e + n^{-C-2} \ge 0$, \rev{where $n^{-C} \le c_e \le n^C$ is the capacity aspect ratio of $G$}. \rev{We note that adding $n^{-C-2}$ to $w_{G'(\ell)}$ is simply so that we can apply \Cref{thm:hst} which requires us to have polynomially-bounded weights on $G'(\ell)$.} With this, we present the dual:
  \begin{align*}
    \max_{\beta, \ell}. & \quad \beta \\
    \text{such that} & \\    
    \forall \lambda \in \Lambda & \quad \sum_{\{u, v\} \in E(G)} c_{\{u, v\}} \cdot \E_{(T, T^G) \sim \calT(\lambda)}\left[ \1{u, v \in V(T)} \cdot w_{G'(\ell)}(T^G_{u, v}) \right] \ge \beta \\
    & \quad \ell \ge 0, \sum_{e \in {E(G)}} \ell_e = 1
  \end{align*}
  
  By inspecting the dual, we see that in order to show that the optimal value of the linear program is at most $\beta$, it is sufficient to show that for every distribution $( \ell_e )_{e \in {E(G)}}$ there exists a $\lambda \in \Lambda$ where $\ell^T M \lambda \le \beta$. To this end, fix any $(\ell_e)_e$ and consider $G'(\ell)$ as defined above. Via \Cref{thm:hst}, there exists a $1$-hop partial embedding distribution $\calT'$ with exclusion probability $\eps$, hop stretch $O(\frac{\log^3 n}{\eps})$, and expected distance stretch $O(\log n \cdot \log \frac{\log n}{\eps})$, i.e.,
  $$\E_{(T, T^G) \sim \calT'}\left[ \1{u, v \in V(T)} \cdot w_{G'(\ell)}(T^G_{u, v}) \right ] \le w_{G'(\ell)}(\{u, v\}) \cdot O(\log n\cdot \log \frac{\log n}{\eps}).$$
  Note that $\calT'$ can be represented as $\calT(\lambda')$ for some $\lambda' \in \Lambda$ since the distribution satisfies the exclusion property and each embedding in the support satisfies the dilation properties (due to the hop stretch). Therefore, for $\lambda = \lambda'$ we have:
  \begin{align*}
    \ell^T M \lambda = & \sum_{\{u, v\} \in E(G)} c_{\{u, v\}} \cdot  \E_{(T, T^G) \sim \calT(\lambda')}\left[ \1{u, v \in V(T)} \cdot w_{G'(\ell)}(T^G_{u, v}) \right] \\ 
    & \le \sum_{\{u, v\} \in E(G)} c_{\{u, v\}} \cdot w_{G'(\ell)}(\{u, v\}) \cdot O(\log n \cdot \log \frac{\log n}{\eps}) \\
    & = O(\log n \cdot \log \frac{\log n}{\eps}) \sum_{\{u, v\} \in E(G)} c_{\{u, v\}} \cdot [\frac{\ell_{\{u, v\}}}{c_{\{u, v\}}} + \rev{\frac{1}{n^{C+2}}}] \\
    & \le O(\log n \cdot \log \frac{\log n}{\eps})[1 + 1] = O(\log n \cdot \log \frac{\log n}{\eps}) .
  \end{align*}

  In other words, we conclude that the optimal value of the linear program is at most $O(\log n \cdot \log \frac{\log n}{\eps})$, showing that there exists a distribution over partial tree embeddings of $G$ that are a $\calD\at{1}$-router with dilation $O(\frac{\log^3 n}{\eps})$ (implied by $\lambda \in \Lambda$), exclusion probability $\eps$ (implied by $\lambda \in \Lambda$), and congestion $O(\log n \cdot \log \frac{\log n}{\eps})$ (optimal linear program value).
\end{proof}

\medskip

Having constructed $\calD\at{1}$-routers, the next and final step is to ``lift'' them into a proper hop-constrained oblivious routing. In this section, we aim only to give an overview, hence we will only show the sampling algorithm and defer arguing about its correctness to \Cref{sec:lifting-d1-to-routings}. \Cref{alg:sampling-hop-constrained-routings} shows how to sample the $h$-hop oblivious routing $R = \{ R_{s, t} \}_{s, t \in V}$ that satisfies the constraints of our main result, \Cref{thmGeneralRouting}. 

\begin{algorithm}
  \caption{Sample a path $p \sim R_{s, t}$, given a graph $G$, $s, t \in V(G)$, and a hop constraint $h \ge 1$.}
  \label{alg:sampling-hop-constrained-routings}
  \begin{algorithmic}[1]%
    \STATE Create a ``completion'' $H = (V, \binom{V}{2}, c_H)$ of $G = (V, E, c_G)$.
    \bindent
    \STATE $c_H(e) := c_G(e)$ if $e \in E$,
    \STATE $c_H(e) := n^{-O(1)} \cdot \min_{e \in e} c_G(e)$ otherwise.
    \eindent
    \STATE Let $\calT_1$ be a $\calD\at{1}$-router on $H$ with exclusion probability $\eps_1 = 1/(4h)$ 
    \STATE Sample $r := O(\log n)$ trees $T_1, T_2, \ldots, T_r \sim \calT_1$ conditioned on $s, t \in V(T_i)$.
    \STATE Assign $q_1 := (T_1)^G_{s, t}, q_2 := (T_2)^G_{s, t}, \ldots, q_r := (T_r)^G_{s, t}$.
    \STATE Let $\calT_2$ be a $\calD\at{1}$-router on $H$ with exclusion probability $\eps_2 = 1 / O(h \log^4 n)$.
    \STATE Sample a tree $(F, F^G) \sim \calT_2$ and let $p := F_{s, t}^G$. \label{line:sample}
    \bindent
    \STATE Simplify $p$ by eliminating all cycles.
    \STATE Repeat the sampling of $F$ and $p$ if $\bigcup_{i=1}^r V(q_i) \not \subseteq V(F)$.
    \STATE Repeat the sampling of $F$ and $p$ if $E(p) \not \subseteq E(G)$.
    \eindent
    \STATE Return $p$.
  \end{algorithmic}
\end{algorithm}

We note that the final routing scheme $R$ is a conditional distribution induced by partial tree embeddings $\calT_2$, conditioned on (1) the sampled tree $F$ containing the nodes of $O(\log n)$ sampled paths from $\calT_1$, and (2) the sampled path using only edges in $G$ (i.e., not using virtual edges constructed during the completion of $G$).

\section{Lifting the $\calD\at{1}$-router to a Hop-Constrained Oblivious Routing}\label{sec:lifting-d1-to-routings}

In this section, we describe and prove how to construct hop-constrained oblivious routings satisfying \Cref{thmGeneralRouting} from $\calD\at{1}$-routers. From a high-level, we first show that $\calD\at{1}$-routers can be used to route other demands $\calD \neq \calD\at{1}$ if one allows for a constant fraction of \rev{\emph{hidden}} failures (\Cref{sec:subflow-routing}). \rev{The failures are hidden in the sense that they adaptively depend on the specific demand $\calD$ and there is no simple way to discern the failed routes from the non-failed ones.} Next, we show a method of ``correcting'' the number of failures down to $n^{-O(1)}$ (i.e., an arbitrarily small polynomial fraction) with a polylogarithmic increase in congestion approximation and hop stretch guarantees (\Cref{sec:correcting}). Finally, we show how to eliminate failures entirely and extend our results to non-complete graphs (\Cref{sec:putting-it-together}).

\subsection{Hop-constrained Subflow Routing}\label{sec:subflow-routing}

In this section, we argue that $\calD\at{1}$-routers are indeed useful for demands $\calD \neq \calD\at{1}$. This is not unexpected: for tree-based routings in the congestion-only setting, a good $\calD\at{1}$-router immediately gives a good oblivious routing (\Cref{lemma:d1-to-all-trees}). However, when dealing with distributions over partial tree embedding, one needs to take special care of failures $\bot$ that can arise when nodes are missing from the partial tree embeddings.

We introduce the concept of \emph{subdistributions}. Suppose that we have a random variable $x$ which, sometimes, produces an unusable result. To model this, we introduce another random variable $y$, which can either be equal to $x$ (in case of success) or be $\bot$ in case of failures. Moreover, the probability of failure is controlled. The distributions of such variables $x$ and $y$ satisfy the following relation.
\begin{definition}
  Let $X$ be a distribution over a set $U$. A distribution $Y$ over $U \cup \{\bot\}$ is a $\gamma$-\textbf{subdistribution} of $X$ if for all $u \in U$ it holds that $\Pr[Y = u] \le \Pr[X = u]$ and $\Pr[Y = \bot] \le \gamma$.
\end{definition}
An equivalent definition of $Y$ being a subdistribution of $X$ is to say that we can construct a probability space with random variables $x \sim X$ and $y \sim Y$ such that $y \in \{ x, \bot \}$.

To simplify notation, in this section we will often conflate a distribution $X$ and a random variable $x \sim X$. Naturally, one has to be careful in doing so since defining random variables requires defining a probability space. Sometimes there is no ambiguity about how to properly formalize the space (e.g., \Cref{def:subflow}, where linearity of expectation makes differences immaterial). However, in places where it matters, we will be careful to make the space clear from the context.

We now define the main concept of this section: subflow routing. Intuitively, a subflow routing is a routing scheme $R = \{ R_{s, t} \}_{s, t \in V(G)}$ where we allow some paths to ``fail''. Similar to $\calD\at{1}$-routers, the paths that fail are not counted towards the congestion. Moreover, these failures can be adaptive to the demand, i.e., they are demand dependent, but the fraction of failures must be tightly controlled by a new parameter $0 < \gamma < 1$.
\begin{definition}\label{def:subflow}
  An \textbf{$h$-hop $\gamma$-subflow routing} with \textbf{congestion approximation} $\alpha \ge 1$ for a graph $G = (V, E)$ is a routing scheme $R = \{R_{u, v}\}_{u, v\in V}$ with the following property. For every demand $\calD : V \times V \to \mathbb{R}_{\ge 0}$ there exists a ``routing scheme with failures'' $R' = R'(\calD) = \{ R'_{s, t} \}_{s, t \in V}$ where $R'_{s, t}$ is an $\gamma$-subdistribution of $R_{s, t}$ and
  $$\cong_G( \sum_{s, t \in V} \E\left[\ \1{R'_{s, t} \neq \bot} \cdot \flow(R'_{s, t}) \cdot \calD_{s, t}\ \right] ) \le \alpha \cdot \opt\at{h}(\calD) .$$
  Additionally, we say that $R$ has \textbf{hop stretch} $\beta$ if for all $s, t \in V$ all paths $p \in \supp(R_{s, t})$ have $\hop(p) \le \beta h$.
\end{definition}

While in this paper we do not focus much on the computational aspects, we will note that the user of an ($h$-hop) $\gamma$-subflow routing cannot differentiate between failures and non-failures (otherwise they could simply route the non-failures, obtain the advertised $\alpha$-approximate congestion for a $1 - \gamma$ fraction of the demands, and then repeat $O_\gamma(\log n)$ times). The crux of this section is a method to \emph{obliviously} boost $\gamma$ down to $n^{-O(1)}$ (\Cref{sec:correcting}). We can then entirely eliminate $\gamma$ (\Cref{sec:putting-it-together}). We emphasize the oblivious part since the final hop-constrained oblivious routing makes no mention of failures (e.g., see \Cref{alg:sampling-hop-constrained-routings} or \Cref{thmGeneralRouting}). We now show that $\calD\at{1}$-routers are indeed $1/2$-subflow routings.

\begin{lemma}[$\calD\at{1}$-routers are $1/2$-subflow routings]\label{lemma:d1-router-is-subflow}
  Let $\calT$ be a $\calD\at{1}$-router on $G$ with dilation $\beta h$, exclusion probability at most $1 / (4h)$, and congestion $\alpha \ge 1$. Let $R^\calT$ be the routing scheme induced by $\calT$ where the distribution $R^\calT_{s, t}$ corresponds to sampling $(T, T^G) \sim \calT$ and returning $T^G_{s, t}$ if $s, t \in V(T)$ and an arbitrary path otherwise. Then $R^\calT$ is an $h$-hop $\frac{1}{2}$-subflow routing with congestion approximation $\alpha$ and hop stretch $\beta$.
\end{lemma}

We remark that our definition of $R^\calT_{s, t}$ sometimes samples a tree $T$ such that $s \not \in V(T)$ or $t \not \in V(T)$, in which case we return an arbitrary path. This is because, by definition, $R^\calT_{s, t}$ must always return a path. However, the definition of subflow routing allows the later analysis to ignore such paths. Specifically, our analysis will draw a $\bot$ in such cases.

\begin{proof}
  For simplicity of notation, we extend the definition of $\flow$ (which maps paths to their flows in $\mathbb{R}_{\ge 0}^{E(G)}$) to $\flow(\bot) = \vec{0}$. Due to this, we can replace $\1{p \neq \bot} \cdot \flow(p) = \flow(p)$.

  Clearly, the dilation property of $\calD\at{1}$-routers implies that each path $p \in \supp(T^G_{s, t})$ has $\hop(p) \le \beta h$, hence the hop stretch property of $R^\calT$ is immediate.

  \medskip

  \rev{\textbf{Proof sketch.} Fix a demand $\calD$ and consider some ($\calD$-dependent) optimal ``witness'' solution $P^*$. Suppose, for simplicity, that each demand pair $(s, t) \in \supp(\calD)$ is routed along a single path and that the paths are edge-disjoint (hence $\opt\at{h}(\calD) = 1$). For some demand pair $(s, t)$ we sample a partial tree $(T, T^G)\sim \calT$ and consider its witness path $P^*_{s, t}$. If all nodes of $P^*_{s, t}$ are in $V(T)$, then we label this a ``success'' and route $s$ to $t$ via $T^G_{s, t}$. On the other hand, if some node is not in $V(T)$, this is a ``failure'' and no congestion is incurred. However, using a union bound and exclusion probability, the probability of failure is at most $\gamma = \frac{1}{4h} \cdot (\hop(P^*_{s, t}) + 1) \le 1/2$. The only remaining thing to argue is the cumulative congestion of all successful routes. If the demand pair $(s, t)$ was successfully routed, then by definition all nodes $\{v_i\}_{i=0}^h$ of the path $P^*_{s, t} = (s = v_0, v_1, \ldots, v_h = t)$ are in $V(T)$. Hence, routing $T^G_{s, t}$ incurs less congestion than cumulatively routing $T^G_{v_0, v_1}, T^G_{v_1, v_2}, \ldots, T^G_{v_{h-1}, v_h}$. In other words, routing $(s, t)$ cannot be worse than routing each edge on the witness path. Performing the same argument over all pairs $(s, t)$ in the demand of $\calD$ it follows that routing the entire demand cannot be worse than routing all edges in the set of witness paths. However, each edge appears at most once in the set of witness paths (since we assumed $\opt\at{h}(\calD) \le 1$), hence all successful routes incur congestion at most $\alpha$---this is exactly the property of $\calD\at{1}$-routers. The argument extends by linearity to larger values of $\opt\at{h}(\calD)$ and convex combinations of $P^*_{s, t}$.}

  \medskip
  
  \textbf{Construction of subdistributions.} Fix a demand $\calD : V(G) \times V(G) \to \mathbb{R}_{\ge 0}$. By definition, there exists a ``witness'' routing scheme $P^* = \{ P^*_{s, t} \}_{s, t \in V(G)}$ that certifies the optimal $h$-hop routing solution. In other words, the support of $P^*_{s, t}$ is over $h$-hop paths connecting $s$ and $t$, and $\cong_G(\sum_{s, t \in V(G)} \E[ \calD_{s, t} \cdot \flow(P^*_{s, t}) ] ) = \opt\at{h}(\calD)$. Equivalently, expanding the definition of $\cong_G$, for every $\{u, v\} \in E(G)$:
  \begin{align}
    \sum_{s, t \in V(G)} \E[ \calD_{s, t} \cdot \flow(P^*_{s, t})_{\set{u, v}} ] \le \opt\at{h}(\calD) \cdot c_{\set{u, v}} . \label{eq:edge-cong}
  \end{align}
  We now construct the collection of subdistributions $\{ R'_{s, t} \}_{s, t}$. First, fix $s, t \in V(G)$. Then independently (of $P^*$) sample $(T, T^G) \sim \calT$. Consider the event $V(P_{s, t}^*) \subseteq V(T)$, namely, that all nodes of a path $p = P_{s, t}^*$ are in $V(T)$, i.e., $V(p) \subseteq V(T)$. If $V(P_{s, t}^*) \subseteq V(T)$ then we assign $R'_{s, t} = T^G_{s, t}$ and otherwise $R'_{s, t} = \bot$. It is clear that $\Pr[R'_{s, t} = p] \le \Pr[R_{s, t} = p]$ for every $p$ and every $s, t$. Furthermore, using the exclusion probability and a union bound, $\Pr[R'_{s,t} = \bot] = \Pr[V(P_{s, t}^*) \subseteq V(T)] \le \frac{\hop(P^*_{s,t}) + 1}{4h} \le \frac{h + 1}{4h} \le \frac{1}{2}$. Therefore, $R'_{s, t}$ is a $\frac{1}{2}$-subdistribution of $R^\calT_{s, t}$.

  \medskip
  
  \textbf{Subdistribution congestion analysis.} For the sake of the analysis, for each $s, t \in V(G)$ we also introduce a random flow variable $C_{s, t} \in \mathbb{R}_{\ge 0}^{E(G)}$ (in the same probability space as above) as follows. In the event $V(P_{s, t}^*) \subseteq V(T)$ we consider the nodes on the path $P_{s, t}^* = (s = v_0, v_1, \ldots, v_{\ell-1}, v_{\ell} = t)$ and assign
  \begin{align*}
    C_{s,t} := \1{ V(P^*_{s, t}) \subseteq V(T) } \cdot \sum_{i=0}^{\ell-1} \flow(T^G_{v_i, v_{i+1}}) \ge  \1{ V(P^*_{s, t}) \subseteq V(T) } \cdot \flow(T^G_{s, t}) .
  \end{align*}
  Note that the right-hand side inequality holds because in any partial tree embedding $\flow(T^G_{a, c}) \le \flow(T^G_{a, b}) + \flow(T^G_{b, c})$ for all $a, b, c \in V(T)$. Therefore,
  \begin{align}
    \sum_{s, t} \E[ \flow(R'_{s, t}) \cdot \calD_{s, t} ] = \sum_{s, t} \E\left[ \1{ V(P^*_{s, t}) \subseteq V(T) } \flow(T^G_{s, t}) \cdot \calD_{s, t} \right] \le \sum_{s, t} \E[ C_{s, t} ] \cdot \calD_{s, t}  . \label{eq:C-lb}
  \end{align}

  On the other hand, we now give an upper bound for $\sum_{s, t} \E[ C_{s, t} \cdot \calD_{s, t} ]$. \rev{We note that in the following, the event $V(P^*_{s, t}) \subseteq V(T)$ implies that all intermediate nodes used by $P^*_{s, t}$ are in $T$.}
  \begin{align*}
    \sum_{s, t}\ & \E[ C_{s, t} \calD_{s, t} ]= \sum_{s, t} \calD_{s, t} \cdot \E\left[ \1{ V(P^*_{s, t}) \subseteq V(T) } \cdot \sum_{i=0}^{\ell-1} \flow(T^G_{v_i, v_{i+1}}) \right] \\
    & = \sum_{s, t} \calD_{s, t} \cdot \E\left[ \1{ V(P^*_{s, t}) \subseteq V(T) } \cdot \sum_{\{u, v\} \in E(G)} \flow(P^*_{s, t})_{\{u, v\}} \cdot \flow(T^G_{u, v}) \right] \allowdisplaybreaks \\
    & = \sum_{\{u, v\} \in E(G)} \E\left[\flow(T^G_{u, v}) \cdot \sum_{s, t} \calD_{s, t} \cdot \1{ V(P^*_{s, t}) \subseteq V(T) } \cdot  \flow(P^*_{s, t})_{\{u, v\}} \right] \allowdisplaybreaks \\
    & \le \sum_{\{u, v\} \in E(G)} \E[\1{ u, v \in V(T) } \cdot \flow(T^G_{u, v})] \cdot \sum_{s, t} \calD_{s, t} \cdot \E\left[ \flow(P^*_{s, t})_{\{u, v\}} \right] \allowdisplaybreaks \\
    & \le \sum_{\{u, v\} \in E(G)} \E[\1{ u, v \in V(T) } \cdot \flow(T^G_{u, v})] \cdot \opt\at{h}(\calD) \cdot c_{\set{u, v}} \qquad \text{\Cref{eq:edge-cong}} \\
    & \rev{= \opt\at{h}(\calD) \cdot \sum_{\{u, v\} \in E(G)} \E[\1{ u, v \in V(T) } \cdot \flow(T^G_{u, v}) \cdot c_{\set{u, v}}]} \\
    & \le \opt\at{h}(\calD) \cdot \alpha \cdot \vec{c_G} \qquad \text{\rev{(Property 3 of \Cref{def:d1-router})}}
  \end{align*}
  Combining the above with \Cref{eq:C-lb} we get that $\cong_G(\sum_{s, t} \E[ \flow(R'_{s, t}) \cdot \calD_{s, t} ]) \allowbreak \le \alpha \cdot \opt\at{h}(\calD)$.
\end{proof}

\subsection{Correcting subflow failures}\label{sec:correcting}

In this section, we show how to drive down the failure bound $\gamma$ to $\gamma^r$ \rev{(e.g., from $1/2$ to $n^{-C}$, for any constant $C > 0$ by setting $r := O(\log n)$)}. Formally, suppose that $R$ is a $\gamma$-subflow routing and fix $s, t \in V(G)$. We sample $r$ paths $\{q_i\}_{i=1}^r$ from $R_{s, t}$ and let $\calT$ be a $\calD_1$-router with exclusion probability at most $( 2 \sum_{i=1}^r \hop(q_i) )^{-1}$. We sample a single tree embedding $(T, T^G) \sim \calT$ conditioned on $V(T)$ containing all nodes of $\{q_i\}_{i=1}^r$, i.e., $\bigcup_{i=1}^r V(q_i) \subseteq V(T)$. The sampled path in our new routing scheme is then $F_{s, t} := T^G_{s, t}$: we claim $F$ is an $\gamma^r$-subflow routing (see \Cref{alg:sampling-hop-constrained-routings}, lines 5--12, ignore lines 9 and 11 which come from \Cref{sec:putting-it-together}). 

\begin{lemma}[Reducing $\gamma$]\label{lemma:gamma-correcting}
  Given an $h$-hop $\gamma$-subflow routing $R$ with hop stretch $\beta$ and congestion approximation $\alpha$ for a (complete) capacitated graph $G$, there exists an $h$-hop $(\gamma^r)$-subflow routing $F$ with hop stretch $O(r \beta \log^3 n)$ and congestion approximation $O(\frac{\alpha}{1 - \gamma} \cdot \log n \cdot \log(r h \beta \log n))$, for every integer $r \ge 2$.
\end{lemma}
\begin{proof}
  We first start with a proof sketch and then show the claim formally. \medskip
  
  \rev{\textbf{Proof sketch.} Suppose $\gamma = 1/2$ and we want to construct an $n^{-O(1)}$-subflow routing. To restate the algorithm, this is accomplished by taking a $1/2$-subflow routing $F$ and sampling $r := O(\log n)$ ``cover paths'' $\{ q_i\}_{i=1}^r$ between each $s, t$. Then, we sample a $\calD\at{1}$-router $(T, \cdot) \sim \calT$ with a (bolstered) exclusion probability $\eps = ( 2 \sum_{i=1}^r \hop(q_i) )^{-1}$, i.e., such that the entire node-set of cover paths (for a fixed $s, t$) appears in $T$ with probability at least $1/2$. Now, consider such a routing from the perspective of some fixed demand $\calD$: each cover path has a (demand-dependent) failure probability $1/2$; we specify that the routing between $s, t$ fails when all cover paths between $s, t$ fail, hence the routing between some pair fails with probability at most $(1/2)^r = n^{-O(1)}$ (proving that the failure probability is boosted). Furthermore, by definition of subflow routings, the non-failed cover paths have congestion competitive with the optimal solution. Therefore, routing all of non-failed cover paths over the $\calD\at{1}$-router will also have competitive congestion with the optimal solution: we can charge the routing between the endpoints $s, t$ on $T$ to routing each edge of all non-failed cover paths on $T$ (since all intermediate nodes are present in the partial tree embedding $T$), which can then in turn be charged to the optimal solution (up to the congestion approximation of the $\calD\at{1}$-routers). This also proves the congestion claim and completes this sketch.}

  \medskip

  \textbf{Notation.} For simplicity of notation, we extend the definition of $\flow$ (which maps paths to their flows in $\mathbb{R}_{\ge 0}^{E(G)}$) to $\flow(\bot) = \vec{0}$. Due to this, we can replace $\1{p \neq \bot} \cdot \flow(p) = \flow(p)$. In the rest of the proof, let $\calT$ be a $\calD\at{1}$-router with exclusion probability $1 / (2 r h \beta)$, dilation $O( r h \beta \log^3 n)$, and congestion $O(\log n \cdot \log (r h \beta \log n) )$ (via \Cref{lemmaDOneRouter}).

  \medskip
  
  \textbf{Construction of the routing scheme $F = \{ F_{u, v} \}_{u, v \in V(G)}$.} Fix $u, v \in V(G)$. We construct $F_{u, v}$ as follows. Independently sample $r$ paths $q_{u, v, 1}, q_{u, v, 2}, \ldots, q_{u, v, r}$ from $R_{u, v}$. Let $S_{u, v} = \bigcup_{i=1}^r V(q_{u, v, i})$ be the set of nodes on the union of the $r$ sampled paths. Note that we have $|S_{u, v}| \le |\{u, v\}| + \sum_{i=1}^r \left(\hop(q_{u, v, i}) - 1\right) \le 2 + r (h \beta - 1) \le r h \beta$. Now, consider the partial tree distribution $(T, T^G)$ from $\calT$ conditioned on $S_{u, v} \subseteq V(T)$. We denote this conditional distribution as ``$\calT \mid S_{u, v} \subseteq V(T)$''. With this notation in place, we set $F_{u, v} = T^G_{u, v}$ where $(T, T^G)$ is (independently) sampled from $\calT \mid S_{u, v} \subseteq V(T)$.

        

  \medskip
  
  \textbf{Hop stretch.} Since $\calT$ has dilation $h \cdot O(r \beta \log^3 n)$ we conclude that the hop stretch of $F$ is $O(r \beta \log^3 n)$.

  \medskip
      
  \textbf{Construction of subdistributions.} Fix a demand $\calD : V(G) \times V(G) \to \mathbb{R}_{\ge 0}$. We now construct the collection of subdistributions $F' = \{ F'_{u, v} \}_{u, v}$. Fix $u, v \in V(G)$. We can reinterpret the construction of $F_{u, v}$ in the following way. The original process independently samples $r$ paths $q_{u,v,1}, \ldots, q_{u,v,r} \sim R_{{u, v}}$. We reinterpret this as sampling $r$ paths $q'_{u,v,1}, \ldots, q'_{u,v,r}$ independently from $R'_{u, v}$, where $R'_{u, v}$ is the natural $\gamma$-subdistribution of $R_{u, v}$ that depends on the demand $\calD$ (i.e., where the failures are bounded in frequency, while the successes have some total congestion). In other words, either $q'_{u,v, j} \gets q_{u,v,j}$ or $q'_{u,v, j} \gets \bot$, where the latter happens with probability $\Pr[R'_{u, v} = \bot] \le \gamma$. We define a new random variable $l'_{u, v}$ to be $q'_{u,v, j}$ where $j = \min\{ j : q'_{u,v,j} \neq \bot \}$; otherwise we define $l'_{u,v} = \bot$ if all $q_{u, v, j} = \bot$ for all $j$. If $l'_{u, v} = \bot$, then we set $F'_{u, v} \gets \bot$. Otherwise $F'_{u, v} \gets F_{u, v}$.

  \medskip
  
  \textbf{Property: $\{F'_{u, v}\}_{u, v}$ are subdistributions of $\{F_{u, v}\}_{u, v}$.} First, it is clear that for $x \neq \bot$ we have $\Pr[F'_{u, v} = x] = \Pr[F_{u, v} = x, l'_{u, v} \neq \bot] \le \Pr[F_{u, v} = x]$. Furthermore, we have that $F'_{u, v} = \bot$ only when all $r$ paths $\{q'_{u, v, j}\}_j$ are sampled as $\bot$, which happens with probability at most $\gamma^r$, therefore $\Pr[F'_{u, v} = \bot] \le \gamma^r$. We conclude that (the distribution of) $F'_{u,v}$ is a $(\gamma^r)$-subdistribution of (the distribution of) $F_{u, v}$.

  \medskip
  
  \textbf{The collection of subdistributions $\{l'_{s, t}\}_{s, t \in V(G)}$ has small congestion.} Each $l'_{s, t}$ defines a distribution over paths connecting $s$ and $t$, or $\bot$. We argue that
  \begin{align}
    \cong_G( \E\left[ \sum_{s, t} \flow(l'_{s, t}) \cdot \calD_{s, t} \right] ) \le \alpha \cdot \frac{1}{1 - \gamma} \cdot \opt\at{h}(\calD) . \label{eq:l-is-witness}
  \end{align}
  To this end, we consider the event $l'_{s, t} \neq \bot$. An equivalent process of sampling $l'_{s, t}$ is the following: sample a path from $R'_{{s, t}}$ and repeat until the path is not $\bot$. Such a rejection sampling is clearly equivalent to sampling from the conditional distribution $R'_{{s, t}} \mid R'_{{s, t}} \neq \bot$. Therefore,
  \begin{align*}
    \E[ \flow(l'_{s, t}) ] & \le (\Pr[R'_{s, t} \neq \bot])^{-1} \cdot \E[ \flow(R'_{s, t}) ] \le \frac{1}{1 - \gamma} \cdot \E[ \flow(R'_{s, t}) ] .
  \end{align*}
  Therefore, $\E[ \sum_{s, t} \flow(l'_{s, t}) \cdot \calD_{s, t} ] \le \frac{1}{1 - \gamma} \cdot \E[ \sum_{s, t} \flow(R'_{s, t}) \cdot \calD_{s, t} ]$. Finally, due to the congestion approximation property of $R'$, we have that $\cong_G(\E[ \sum_{s, t} \flow(R'_{s, t}) \cdot \calD_{s, t} ]) \le \alpha \cdot \frac{1}{1 - \gamma} \cdot \opt\at{h}(\calD)$.


  \medskip
  
  \textbf{Property: congestion of subdistributions.} We analyze $\E[ \sum_{s, t} \flow(F'_{s, t}) \cdot \calD_{s, t} ]$. First, we introduce some notation: we remind the reader that $F_{s, t}$ is drawn from (an embedding from) $\calT \mid S_{u,v} \subseteq V(T)$. On the other hand, we define $(T, T^G)$ to be an independent random variable drawn from $\calT$. Finally, we define an event $EV := \{ l'_{s, t} \neq \bot, S_{u,v} \subseteq V(T) \}$.
  \begin{align}
    \E[ \flow(F'_{s, t}) ] & = \E[ \1{l'_{s,t} \neq \bot} \cdot \flow(F_{s, t}) ]  \nonumber \\ 
                           & \le \E\left[ \1{ EV } \cdot \flow(T^G_{s, t}) \right] / \Pr[S_{u,v} \subseteq V(T)] \nonumber \\
                           & \le 2 \cdot \E\left[ \1{ EV } \cdot \flow(T^G_{s, t}) \right]  \nonumber 
  \end{align}
  The last inequality follows from $\Pr[S_{u, v} \subseteq V(T)] \ge 1/2$ which is a result of a simple union bound, $|S_{u, v}| \le r h \beta$, and the exclusion probability of $\calT$ being at most $1 / (2 r h \beta)$.
  
  For each $s, t \in V(G)$ we introduce a random flow variable $C_{s, t} \in \mathbb{R}_{\ge 0}^{E(G)}$ as follows. In the event $EV$ we consider the nodes on the path $l'_{s, t} = (s = v_0, v_1, \ldots, v_{\ell-1}, v_{\ell} = t)$ and assign
  \begin{align}
    C_{s,t} := \1{ EV } \cdot \sum_{i=0}^{\ell-1} \flow(T^G_{v_i, v_{i+1}}) \ge  \1{ EV } \cdot \flow(T^G_{s, t}) . \label{eq:cst-vs-tst}
  \end{align}
  Note that the right-hand side inequality of \Cref{eq:cst-vs-tst} holds because in any partial tree embedding $\flow(T^G_{a, c}) \le \flow(T^G_{a, b}) + \flow(T^G_{b, c})$ for all $a, b, c \in V(T)$. Therefore,
  \begin{align}
    \sum_{s, t} \E[ \flow(F'_{s, t}) \cdot \calD_{s, t} ] \le 2 \sum_{s, t} \E\left[ \1{ EV } \cdot \flow(T^G_{s, t}) \cdot \calD_{s, t} \right] \le 2 \sum_{s, t} \E[ C_{s, t} ] \cdot \calD_{s, t} . \label{eq:C-lb-2}
  \end{align}
  On the other hand, we now give an upper bound for $\sum_{s, t} \E[ C_{s, t} \cdot \calD_{s, t} ]$.
  \begin{align*}
    \sum_{s, t} & \E[ C_{s, t} \calD_{s, t} ] = \sum_{s, t} \calD_{s, t} \cdot \E\left[ \1{ EV } \cdot \sum_{i=0}^{\ell-1} \flow(T^G_{v_i, v_{i+1}}) \right] \\
    & = \sum_{s, t} \calD_{s, t} \cdot \E\left[ \1{ EV } \cdot \sum_{\{u, v\} \in E(G)} \flow(l'_{s, t})_{\{u, v\}} \cdot \flow(T^G_{u, v}) \right] \\
    & = \sum_{\{u, v\} \in E(G)} \E\left[\flow(T^G_{u, v}) \cdot \sum_{s, t} \calD_{s, t} \cdot \1{ EV } \cdot  \flow(l'_{s, t})_{\{u, v\}} \right] \\
    & \le \sum_{\{u, v\} \in E(G)} \E[\1{u, v \in V(T) } \cdot \flow(T^G_{u, v})] \cdot \sum_{s, t} \calD_{s, t} \cdot \E\left[ \flow(l'_{s, t})_{\{u, v\}} \right] \\
    & \le \sum_{\{u, v\} \in E(G)} \E[\1{ u, v \in V(T) } \cdot \flow(T^G_{u, v})] \cdot \frac{\alpha}{1 - \gamma} \cdot \opt\at{h}(\calD) \cdot c_{\set{u, v}} \quad \text{(Eq.~\ref{eq:l-is-witness})} \\
    & = \frac{\alpha}{1 - \gamma} \cdot \opt\at{h}(\calD) \cdot \sum_{\{u, v\} \in E(G)} c_{\set{u, v}} \cdot \E[ \1{ u, v \in V(T) } \cdot \flow(T^G_{u, v})] \\    
    & = \frac{\alpha}{1 - \gamma} \cdot \opt\at{h}(\calD) \cdot O(\log n \cdot \log(rh\beta \log n)) \cdot \vec{c_G} \qquad \text{($\calD\at{1}$-router properties)}
  \end{align*}
  Combining the above with \Cref{eq:C-lb-2} we get that $\cong_G(\sum_{s, t} \E[ \flow(F'_{s, t}) \cdot \calD_{s, t} ]) \le O(\frac{\alpha}{1 - \gamma} \cdot \log n \cdot \log(r h \beta \log n)) \cdot \opt\at{h}(\calD)$.
\end{proof}

We combine the results that we developed so far.

\begin{corollary}\label{corollary:good-subflow-distrib}
  For every (complete) capacitated graph $G = (V, E, c)$ \rev{with poly\-nomially-bounded capacities} and every $h \ge 1, r = O(1)$, there exists an $h$-hop $(n^{-r})$-subflow routing for $G$ with congestion approximation $O(\log^2 n \cdot \log^2 (h \log n))$ and hop stretch $O(\log^7 n)$.
\end{corollary}
\begin{proof}
  There exists a $\calD\at{1}$-router $\calT_1$ on $G$ with dilation $O(h \log^3 n)$, exclusion probability $1 / (4h)$, and congestion $O(\log n \log ( h \log n ) )$ (via \Cref{lemmaDOneRouter}).

  Define $R := \{ R_{s, t} \}_{s, t \in V(G)}$ with $R_{s, t} := T^G_{s, t}$, where $(T, T^G) \sim \calT_1$. Applying \Cref{lemma:d1-router-is-subflow}), we conclude that $R$ is an $h$-hop $\frac{1}{2}$-subflow routing with congestion $O(\log n \log ( h \log n ) )$ and hop stretch $O(\log^3 n)$.

  Finally, correcting the failures in the subflow routing via \Cref{lemma:gamma-correcting} by setting $r_{(\Cref{lemma:gamma-correcting})} := r \log_2 n = O(\log n)$, we construct an $h$-hop $(n^{-r})$-subflow routing with congestion approximation $O(\log^2 n \cdot \allowbreak\log^2 (h \log n))$, hop stretch $O(\log^7 n)$.
\end{proof}

\subsection{Putting it together: Non-subflow routing on general graphs}\label{sec:putting-it-together}

In this section, we prove our main result by combining all of the above. On a high-level, the main technical contribution of this section is to (1) completely eliminate failures, and to (2) extend the results from complete capacitated graphs to general capacitated graphs. However, both of these issues can be resolved in the following way.
\begin{enumerate}
\item A $n^{-C}$-fraction of failures (for any constant $C > 0$) can readily be ignored since they contribute an insignificant amount to the congestion.
\item We can ``complete'' a general capacitated graph into its completed counterpart by converting ``non-edges'' to edges of sufficiently small capacity $n^{-O(1)}$. We construct the hop-constrained oblivious routing $R$ on the completed graph. Note that $R$ is supported (with small probability) over non-edges of the original graph. However, we can easily condition on these paths not using non-edges, which will only insignificantly increase the congestion.
\end{enumerate}

\thmGeneralRouting*

\begin{proof}
  For simplicity of notation, we extend the definition of $\flow$ (which maps paths to their flows in $\mathbb{R}_{\ge 0}^{E(G)}$) to $\flow(\bot) = \vec{0}$. Due to this, we can replace $\1{p \neq \bot} \cdot \flow(p) = \flow(p)$. Furthermore, let $c_{\min} = \min_{e \in E(G)} c_G(e)$, $c_{\max} = \max_{e \in E(G)} e_G(e)$, and let $C = O(1)$ be a sufficiently large constant. We remind the reader that because the capacities are polynomially bounded we have $c_{\min} \ge n^{- O(1)}$ and $c_{\max} \le n^{O(1)}$.

  \medskip
  
  \textbf{Completing the graph.} We first construct a ``completed'' capacitated graph $H = (V(G),\allowbreak \binom{V(G)}{2},\allowbreak c_{H})$ where $c_{H}(e) := c_G(e)$ if $e \in E(G)$, or $c_{H}(e) := c_{\min} \cdot n^{-C}$ if $e \in E(H) \setminus E(G)$. Due to this choice, for any demand $\calD$, we have that $\opt_{H}\at{h}(\calD) \ge \frac{1}{n^2 c_{\max}} \cdot \sum_{s, t} \calD_{s, t}$ since we are pushing $\sum_{s, t} \calD_{s, t}$ units of flow across at most $n^2$ edges of capacity of at most $c_{\max}$. Therefore, we conclude that $\opt_{H}\at{h}(\calD) \ge n^{-O(1)} \cdot \sum_{s, t} \calD_{s, t}$.

  For the rest of the proof let $R = \{ R_{s, t} \}_{s, t \in V(G)}$ be an $h$-hop $(n^{-2C})$-subflow routing on $H$ with congestion approximation $\alpha := O(\log^2 n \cdot \log^2 (h \log n))$ and hop stretch $O(\log^7 n)$. Furthermore, we can assume that the support of $R_{s,t}$ is over simple paths since we can always simplify each path without increasing the congestion.

  \medskip
    
  \textbf{$R$ has good congestion approximation on $H$ without subflows.} Fix a demand $\calD$ and let $R' = \{ R'_{s, t} \}_{s, t \in V(G)}$ be the collection of $(n^{-2C})$-subdistributions of $R = \{ R_{s, t} \}_{s, t}$ with respect to $\calD$. We also denote by $R_{s,t}$ and $R'_{s, t}$ the random variables (drawn from the distribution of the same name), coupled so that $R'_{s,t} \in \{ R_{s,t}, \bot \}$. With this notation, we now show that the routing scheme $R$ achieves a good congestion approximation on (all) $\calD$.
  \begin{align}
    & \cong_{H}(\sum_{s, t} \E[ \flow_H(R_{s, t}) \cdot \calD_{s, t} ]) \nonumber \\     
    & \le \cong_{H}\left(\sum_{s, t} \E[ \flow_H(R'_{s, t}) \calD_{s, t} ] \right) + \cong_{H}\left( \sum_{s, t} \E[ \1{R'_{s, t} = \bot} \flow_H(R_{s, t}) \calD_{s, t} ] \right ) \nonumber \\
    & \le \alpha \cdot \opt\at{h}_{H}(\calD) + \cong_{H}\left( \sum_{s, t} \E[ \1{R'_{s, t} = \bot} \flow_H(R_{s, t}) \cdot \calD_{s, t} ] \right ) \nonumber \\
    & \le \alpha \cdot \opt\at{h}_{H}(\calD) + \Pr[ R'_{s, t} = \bot ] \cdot (n^C / c_{\min}) \cdot \sum_{s, t} \calD_{s, t} \label{eq:bad-edges-contrib}  \\
    & \le \alpha \cdot \opt\at{h}_{H}(\calD) + n^{-2C} \cdot \frac{n^C}{n^{O(1)}} \cdot n^{O(1)} \cdot \sum_{s, t} \calD_{s, t} \nonumber  \\
    & \le \alpha \cdot \opt\at{h}_{H}(\calD) + \opt\at{h}_{H}(\calD) \label{eq:bad-stuff-at-most-opt}  \\
    & \le 2 \alpha \cdot \opt\at{h}_{H}(\calD) \nonumber
  \end{align}
  \Cref{eq:bad-edges-contrib} holds because each ``bad'' path $R_{s, t}$ (i.e., $R_{s,t}$ when $R'_{s,t} = \bot$) incurs at most $(\min_{e \in E(H)}\allowbreak c_H(e))^{-1} \le (c_{\min} \cdot n^{-C})^{-1}$ congestion per each unit of demand (remember that we can assume paths are simple). \Cref{eq:bad-stuff-at-most-opt} holds for sufficiently large $C$ because $n^{-O(1)} \cdot \sum_{s, t} \calD_{s, t} \le \opt_{H}\at{h}(\calD)$ (as argued before). With this calculation, we conclude that $R$ is an $h$-hop oblivious routing for $H$ with hop stretch $O(\log^7 n)$ and congestion approximation $O(\alpha) = O(\log^2 n \cdot \log^2 (h \log n))$.

  \medskip

  \textbf{Routing scheme for the general graph.} We now adapt the routing scheme $R = \{ R_{s, t} \}_{s, t \in V(G)}$, which is defined on the completed graph $H$, to a new routing scheme $F = \{ F_{s, t} \}_{s, t \in V(G)}$ which is valid on the original (general) graph $G$. We simply define $F_{s, t}$ be $R_{s, t}$ conditioned on the sampled path traversing only edges in $E(G)$. We denote this conditional distribution with $F_{s,t} = R_{s, t} \mid R_{s, t} \subseteq G$.

  Fix $s, t \in V(G)$ such that there exists an $h$-hop path in $G$ between $s$ and $t$. Our aim is to bound the probability that $R_{s, t}$ uses edges not in $E(G)$. Let $\calD\at{s, t}$ be the demand that has a single request between $s$ and $t$; i.e., $\calD\at{s, t}_{u, v} := \1{(u,v) = (s,t)}$. Since there is an $h$-hop path between $s$ and $t$ with edges of capacities at least $c_{\min}$ we have $\opt_G\at{h}(\calD\at{s, t}) \le \frac{1}{c_{\min}} \le n^{O(1)}$. Fix an edge $e \in E(H)\setminus E(G)$, i.e., that does not exist in $G$. By the assumption that $R$ is an $h$-hop oblivious routing for $H$ with congestion stretch $O(\alpha) \le n^{O(1)}$, its routing of $\calD\at{s, t}$ on $H$ is $n^{O(1)}$-competitive, therefore:
  \begin{align*}
    \Pr[e \in R_{s, t}] \le \E[ \flow(R_{s, t}, e) ] \le n^{O(1)} \cdot \opt_{H}\at{h}(\calD) \cdot c_{H}(e) \le n^{O(1) - C}.
  \end{align*}
  Union-bounding over all $E(H) \setminus E(G)$ we have that (for a sufficiently large $C$)
  \begin{align*}
    \Pr[R_{s,t} \not\subseteq G] \le |E(H) \setminus E(G)| \Pr[e \in R_{s, t}] \le n^2 \cdot n^{O(1) - C} \le 1/2 .
  \end{align*}

  Fix an arbitrary demand $\calD$. If there exists $s, t \in V(G)$ such that $\calD_{s, t} > 0$, but there is no $h$-hop path between them, then $\opt_G\at{h}(\calD) = \infty$ and the claim is trivial. If this is not the case, we bound the congestion of the routing scheme $F_{s,t}$ (which is conditioned on going only over edges in $G$):
  \begin{align}
    & \cong_G\left(\sum_{s, t} \E[ \flow_G( F_{s, t} ) \cdot \calD_{s, t} ]\right) \nonumber \\
    & \le \cong_G\left(\sum_{s, t} \E[ \1{R_{s, t} \subseteq G} \cdot \flow_G(R_{s, t}) \cdot \calD_{s, t} ] \right) / \Pr[R_{s, t} \subseteq G] \nonumber \\
    & \le 2 \cdot \cong_H\left( \sum_{s, t} \E[ \flow_H(R_{s,t}) \cdot \calD_{s, t} ] \right ) \label{eq:G-only} \\
    & \le O(\alpha) \cdot \opt_H\at{h}(\calD) \nonumber \\
    & \le O(\alpha) \cdot \opt_G\at{h}(\calD) \label{eq:GH-final} .
  \end{align}
  \Cref{eq:G-only} follows because all paths in the preceding equation go only over edges in $G$; both $G$ and $H$ agree on the capacities of such edges. \Cref{eq:GH-final} follows from $c_H(e) \ge c_G(e)$, hence $\opt_H\at{h}(\calD) \le \opt_G\at{h}(\calD)$. We conclude that $F$ is an $h$-hop oblivious routing on $G$ with hop stretch $O(\log^7 n)$ and congestion approximation $O(\log^2 n \cdot \log^2 (h \log n))$ for all demands $\calD$.
 \end{proof}

\section{Computational Aspects}

This paper primarily focuses on the existence of hop-constrained oblivious routings without talking about how to efficiently construct them. However, it is relatively straightforward to give a randomized construction of such routings in polynomial time. Examining \Cref{alg:sampling-hop-constrained-routings} that constructs these routings, we observe that all of the steps involved are straightforward to implement in polynomial time except, perhaps, constructing $\calD\at{1}$-routers. Moreover, the existence of $\calD_1$-routers is proven via strong duality, making it less clear how to make the result algorithmic. However, this step can be made algorithmic---we can sample $\calD_1$-routers in polynomial time using the standard technique of multiplicative weights~\cite{arora2012multiplicative} in the same way the R\"acke oblivious routing constructions are algorithmic~\cite{racke2008optimal}.

\rev{Furthermore, very recent work by Haeupler, R\"acke, and Ghaffari~\cite{hopexpander2022} has shown that $h$-hop oblivious routing distributions can be constructed in $\poly(h) \cdot m^{1+o(1)}$.}



\appendix
\section{Proof of \Cref{lemma:d1-to-all-trees}}\label{sec:proof-lemma-d1-to-all-trees}

\lemmaDToAllTrees*
\begin{proof}[Proof of \Cref{lemma:d1-to-all-trees}] 
  We prove the claim in this paper for completeness. However, we note that the claim is implicit in, e.g., \cite{racke2008optimal}. Claim 3 in Section 2 of \cite{racke2008optimal} gives a definition of expected relative load $\alpha$ of a distribution over complete tree embeddings. The definition is equivalent to saying that for each edge $e$ the expected amount of flow routed over $e$ when routing the $\calD\at{1}$ demand is at most $\alpha \cdot c_e$. The Subsection titled ``Oblivious Routing'' of Section 3 proves that routing any set of demands $\calD$ over a distribution with expected relative load $\alpha$ implies that the achieved routing has congestion approximation at most $\alpha$, as required.


  \rev{We now prove the claim. First, let $I(s, t)$ be the unit demand between $s, t \in V$, i.e., $I({s, t})_{x, y} = \1{ \{s, t\} = \{x, y\} }$. For any \emph{tree-based} routing scheme $R^\calT$, and any path $p = (s = p_0, p_1, \ldots, p_{\hop(p)} = t)$ we observe that
  $$\flow(I(s, t), R^\calT) \le \sum_{i=0}^{\hop(p) - 1} \flow(I(p_i, p_{i+1}), R^\calT) .$$
  Note that such claims do not hold for general routing schemes. Consequently, for any distribution $\calP$ over paths between $s$ and $t$ we have that
  $$\flow(I(s, t), R^\calT) \le \E_{p \sim \calP}\left[ \sum_{i=0}^{\hop(p) - 1} \flow(I(p_i, p_{i+1}), R^{\calT}) \right] .$$

  Given an arbitrary demand $\calD$, suppose the optimal routing scheme $R(\calD)$ achieves the optimal value $\opt\at{\infty}(\calD)$. In other words, for each edge $\{u, v\} \in E$ we have that
  $$\sum_{s, t \in V} \calD_{s, t}\ \E_{p \sim R(\calD)_{s, t}} \sum_{i=0}^{\hop(p) - 1} \1{\{u, v\} = \{p_i, p_{i+1}\}} \le \opt\at{\infty}(\calD) \cdot c_{\{u, v\}}.$$ We now have:
  \begin{align*}
    \flow(\calD, R^\calT) & = \sum_{s, t \in V} \calD_{s, t}\ \flow( I(s, t), R^{\calT} ) \\
                          & = \sum_{s, t \in V} \calD_{s, t}\ \E_{p \sim R(\calD)_{s, t}}\ \flow( \sum_{i=0}^{\hop(p) - 1} I(p_i, p_{i+1}), R^{\calT} ) \allowdisplaybreaks \\
                          & = \sum_{\{u, v\} \in E} \flow( I(u, v), R^{\calT} )\ \cdot \\
                          & \qquad \cdot \left( \sum_{s, t \in V} \calD_{s, t}\ \E_{p \sim R(\calD)_{s, t}} \sum_{i=0}^{\hop(p) - 1} \1{\{u, v\} = \{p_i, p_{i+1}\}} \right ) \allowdisplaybreaks \\
                          & \le \sum_{\{u, v\} \in E} \flow( I(u, v), R^{\calT} )\ \cdot \opt\at{\infty}(\calD) \cdot c_{\{u, v\}}  \\
                          & = \flow( \calD\at{1}, R^{\calT} )\ \cdot \opt\at{\infty}(\calD) \\
                          & \le \alpha \cdot \vec{c} \cdot \opt\at{\infty}(\calD).
  \end{align*}
  In other words, $\cong_G(\calD, R^{\calT}) \le \alpha \cdot \opt\at{\infty}(\calD)$, as required.} 
\end{proof}

\bibliographystyle{alpha}
\bibliography{refs}

\end{document}